\documentclass[a4paper, 11pt]{article}
\usepackage{hidden_gem}
\title{Lower Bounds on Tree Covers}
\date{}
\author{
Yu Chen
\and
Zihan Tan
\and
Hangyu Xu}

\newcommand{\tset}{{\mathcal{T}}}

\newcommand{\Esp}{E_{\mathrm{sp}}}
\newcommand{\conv}{ {\mathrm{conv}} }

\newcommand{\myparskip}{3pt}
\parskip \myparskip
\begin{document}

\begin{titlepage}
	
\title{Lower Bounds on Tree Covers}

\author{Yu Chen\thanks{National University of Singapore. Email: {\tt yu.chen@nus.edu.sg}.} \and Zihan Tan\thanks{University of Minnesota. Email: {\tt zihantan1993@gmail.com}.} \and Hangyu Xu\thanks{University of Science and Technology of China. Email: {\tt hangyuxu@mail.ustc.edu.cn}.}} 
	
	\maketitle

\begin{abstract}
Given an $n$-point metric space $(X,d_X)$, a tree cover $\mathcal{T}$ is a set of $|\mathcal{T}|=k$ trees on $X$ such that every pair of vertices in $X$ has a low-distortion path in one of the trees in $\mathcal{T}$. Tree covers have been playing a crucial role in graph algorithms for decades, and the research focus is the construction of tree covers with small size $k$ and distortion.

When $k=1$, the best distortion is known to be $\Theta(n)$. For a constant $k\ge 2$, the best distortion upper bound is $\tilde O(n^{\frac 1 k})$ and the strongest lower bound is $\Omega(\log_k n)$, leaving a gap to be closed. In this paper, we improve the lower bound to $\Omega(n^{\frac{1}{2^{k-1}}})$.

Our proof is a novel analysis on a structurally simple grid-like graph, which utilizes some combinatorial fixed-point theorems. We believe that they will prove useful for analyzing other tree-like data structures as well.
\end{abstract}
\end{titlepage}

\thispagestyle{empty}

\clearpage
\pagenumbering{arabic}


\section{Introduction}

Given a metric space $(X,d_X)$ on $|X|=n$ points, an \emph{$(\alpha,k)$-tree cover} is a collection $\tset=\set{T_1,\ldots,T_k}$ of $k$ (edge-weighted trees) on $X$, such that every pair $x,y$ of vertices in $X$,
\begin{itemize}
    \item for each $1\le i\le k$, $\dist_{T_i}(x,y)\ge d_X(x,y)$ (that is, each $T_i$ is a \emph{dominating tree}); and
    \item there exists an $1\le i\le k$, such that $\dist_{T_i}(x,y)\le \alpha\cdot d_X(x,y)$. 
\end{itemize}
The integer $k=|\tset|$ is called the \emph{size} of $\tset$ and the real number $\alpha\ge 1$ is called the \emph{distortion} or \emph{stretch} of $\tset$.
The goal is to construct tree covers with small size and distortion.

As a natural and instrumental approach to approximating metrics by trees, tree cover has been studied extensively and has found numerous algorithmic applications, for example, approximating low-dimensional Euclidean spaces \cite{arya1995euclidean,chang2024optimal}, routing
\cite{awerbuch1992routing,gupta2001traveling,awerbuch2002buffer,chan2016hierarchical}, network design \cite{gupta2006oblivious}, distance oracles \cite{mendel2007ramsey}, Ramsey tree covers \cite{bartal2003metric,bartal2022covering}, Steiner Point Removal \cite{chang2023covering,chang2023resolving,chang2024shortcut}, and many others.

Despite much attention and effort, several fundamental questions regarding the size-distortion trade-off of tree covers remain unsolved, especially in the regime where $k$ is as small as a constant. When $k=1$, the best distortion is known to be $\Theta(n)$: the upper bound is achieved by the minimum spanning tree, and the lower bound is proved by the unweighted $n$-cycle \cite{rabinovich1998lower}.
When $k$ is a small constant, the best upper bound, due to Bartal, Fandina, and Neiman \cite{bartal2022covering}, states that the stretch $\tilde O(n^{1/k})$ is achievable (actually by a more restricted type of tree covers, more discussion in \Cref{subsec: result}). On the other hand, the best lower bound remains $\Omega(\log_k n)$, witnessed by high-girth graphs \cite{bollobas2004extremal}, leaving a big gap.
Closing this gap is viewed as one of the major interesting open problems on tree covers: 
\begin{itemize}
\item ``Is there
a better lower bound? Or perhaps there exists a tree cover with $O(1)$ trees and distortion
$O(\log n)$? What can one say about a tree cover with only 2 trees?'' -- Bartal, Fandina, and Neiman in \cite{bartal2022covering}.
\item ``A very interesting open question: what is the smallest stretch achieved with $2$ trees?'' -- Le\footnote{at Sublinear Graph Simplification workshop at Simons Institute, video recording can be accessed at \url{https://simons.berkeley.edu/talks/hung-le-university-massachusetts-amherst-2024-08-01}}.
\end{itemize}

\subsection{Our result}
\label{subsec: result}

In this paper, we make progress towards closing this gap. Our main result is the following theorem.

\begin{theorem}
\label{thm: main}
For every $k\ge 1$, there is an $n$-point metric, such that any size-$k$ tree cover has stretch $\Omega_k(n^{\frac{1}{2^{k-1}}})$.
\end{theorem}
Specifically, when $k=2$, our result implies a lower bound of $\Omega(\sqrt{n})$, nearly matching the upper bound of $\tilde O(\sqrt{n})$ in \cite{bartal2022covering}. For constant $k$, we show a polynomial lower bound, refuting the possibility of $\poly\log n$ stretch. Therefore, \Cref{thm: main} provides an answer to the above questions by Bartal, Fandina, Neiman, and Le.

For proving \Cref{thm: main}, we employ a novel approach. At a high-level, we construct a grid-like structure, compute its triangulation, label each point according to its location in the tree cover, and derive the stretch lower bound by Tucker's Lemma. Our proof reveals a close connection between tree covers and the celebrated fixed-point theorems in combinatorial topology, and we believe that this connection will be useful in analyzing tree covers in different settings and other tree-like data structures.

The stretch upper bound of $\tilde O(n^{1/k})$ in \cite{bartal2022covering} is actually achieved by Ramsey tree covers, a restricted type of tree covers that require each vertex to choose a home tree to preserve its distances to all other vertices (formal definition in \Cref{sec: ramsey}). They also showed a stretch lower bound of $\Omega(n^{1/k})$ for Ramsey tree covers.
As a side result, we give an alternative proof of this near-optimal lower bound. Our hard instance is a high-dimensional cyclic grid, which is structurally simpler than the recursive construction in \cite{bartal2022covering}. Our proof is similar to the proof of \Cref{thm: main} and also makes use of a variant of Tucker's Lemma.

\section{Technical Overview}

In this section, we provide a high-level overview of our techniques, with the focus on the connection between tree covers and combinatorial topology.

\subsection*{Case $k=1$: examining a tree-based tour on a cycle}
\label{sec: case k=1}

We start by reviewing the proof of the $\Omega(n)$ distortion lower bound for the case $k=1$. We present a proof sketch similar to the one in \cite{rabinovich1998lower}, and an alternative proof can be found in \cite{gupta2001steiner}.

Consider an unweighted cycle $(0,1,2,\ldots,n-1,0)$ and a dominating tree $T$ on the same set of vertices. We may assume without loss of generality that for each edge $(i,j)$ in $T$ (with $i<j$), its length/weight is the cycle-distance between vertices $i$ and $j$, namely $\dist_T(i,j)=\min\set{j-i,n-(j-i)}$, as otherwise we can decrease the length of this edge without hurting the distortion of $T$. To reveal the topological difference between a tree and a cycle, we ``interpolate'' their edges as follows.

\begin{enumerate}
\item For each edge $(i,j)$ in $T$, we subdivide it by $\dist_T(i,j)-1$ new nodes (called \emph{\color{red} interpolating nodes}), which are then labeled by consecutive integers connecting $i,j$ in the natural order. For example, edge $(3,6)\to (3,{\color{red}4},{\color{red}5},6)$, and edge $(1,n-3)\to (1,{\color{red}0},{\color{red}n-1},{\color{red}n-2},n-3)$. Edges $(3,{\color{red}4}),({\color{red}4},{\color{red}5})({\color{red}5},6)$ are called \emph{interpolating edges}. See \Cref{fig: 1tree} for an illustration. 
\item For each $i$, we define $\Pi_i$ as the sequence of all nodes (including the old ones and the interpolating ones) we encounter while walking along the path from the old $i$ to the old $i+1$ in $T$. Observe that the tree-distance between old nodes $i,i+1$ is $|\Pi_i|$, the number of interpolating edges in $\Pi_i$.
Here $\Pi_{n-1}$ connects the old node $n-1$ to the old node $0$ in the tree.
\item
Finally, we define $\Pi$ as the circular concatenation of sequences $\Pi_0,\Pi_1\ldots,\Pi_{n-1}$, identifying, for each $i$, the last node of $\Pi_i$ and the first node of $\Pi_{i+1}$ (both being the old $i+1$), so $\Pi$ is a closed circular tour.
\end{enumerate}

\begin{figure}[h]
	\centering
	\subfigure[The interpolated tree: old nodes (blue) and interpolating nodes (red). $\Pi_5$ illustrated in yellow dashed line.]
	{\scalebox{0.104}{\includegraphics{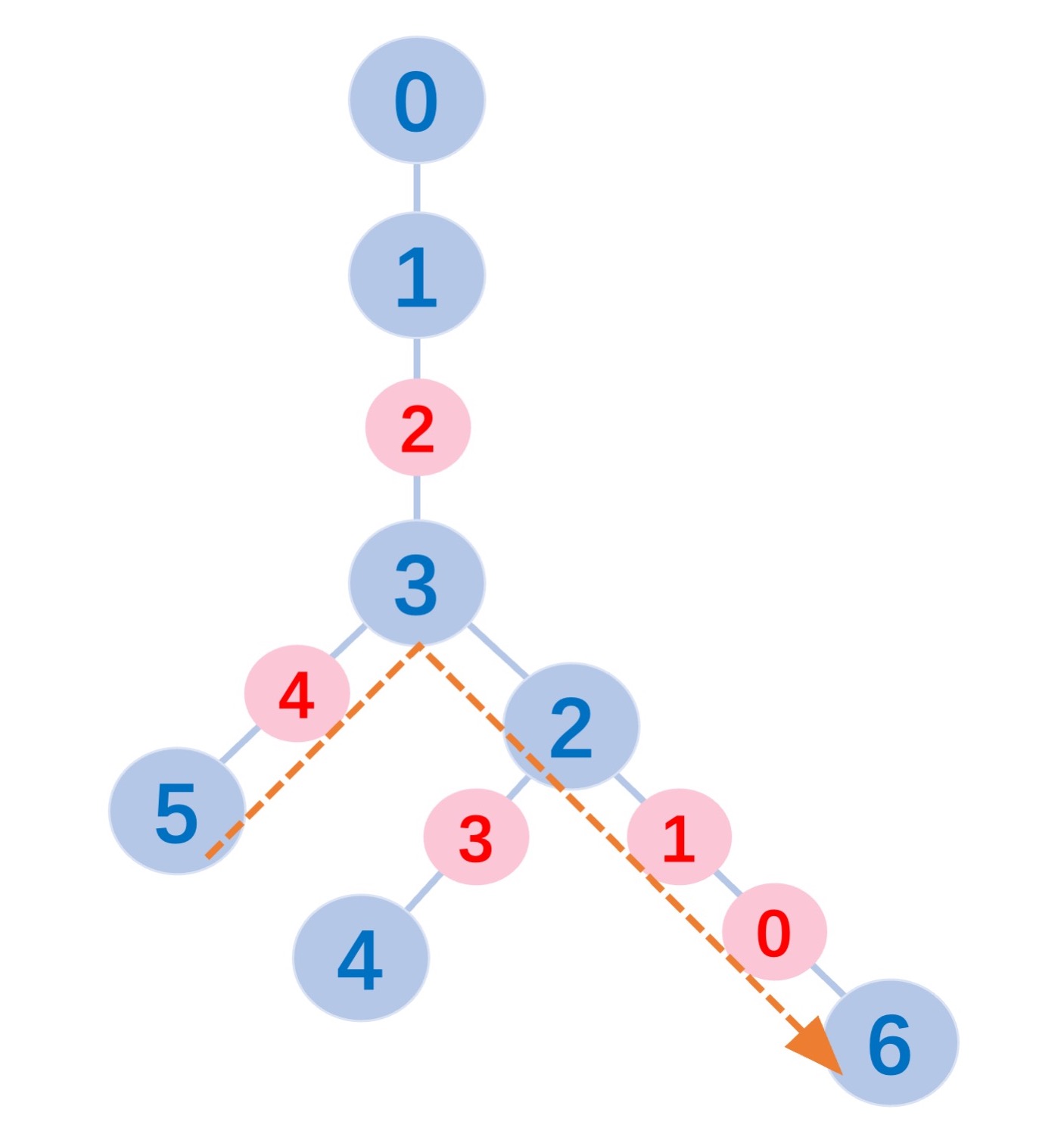}}}
	\hspace{0.2cm}
	\subfigure[The sequences $\Pi_i$'s from the interpolated tree on the left.]
	{
		\scalebox{0.104}{\includegraphics{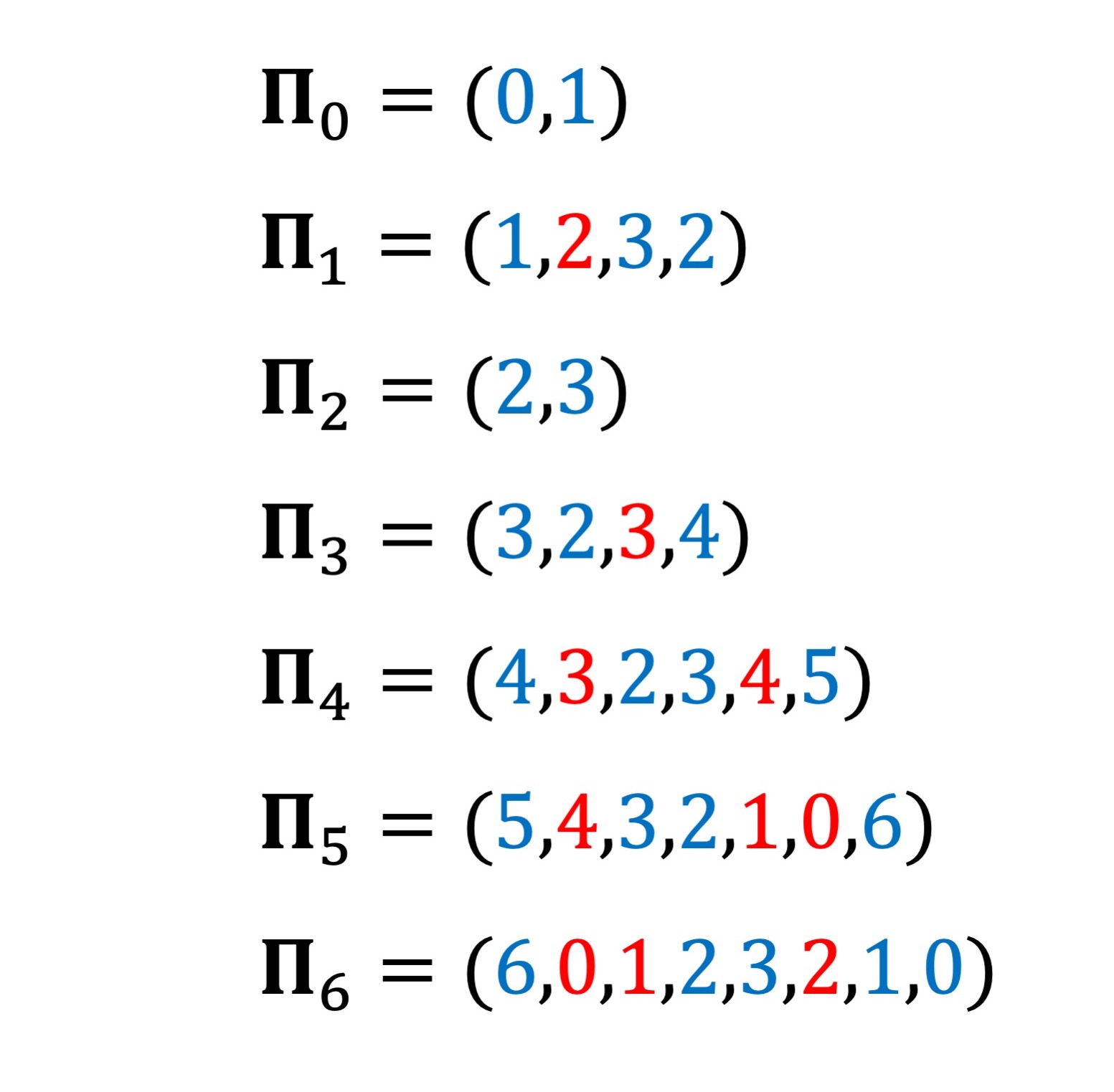}}}
        \hspace{0.2cm}
	\subfigure[$\Pi_5$ is the only bad $\Pi_i$, as it goes from $5$ to $6$ around the cycle, i.e., from the ``far side''.]
	{
		\scalebox{0.104}{\includegraphics{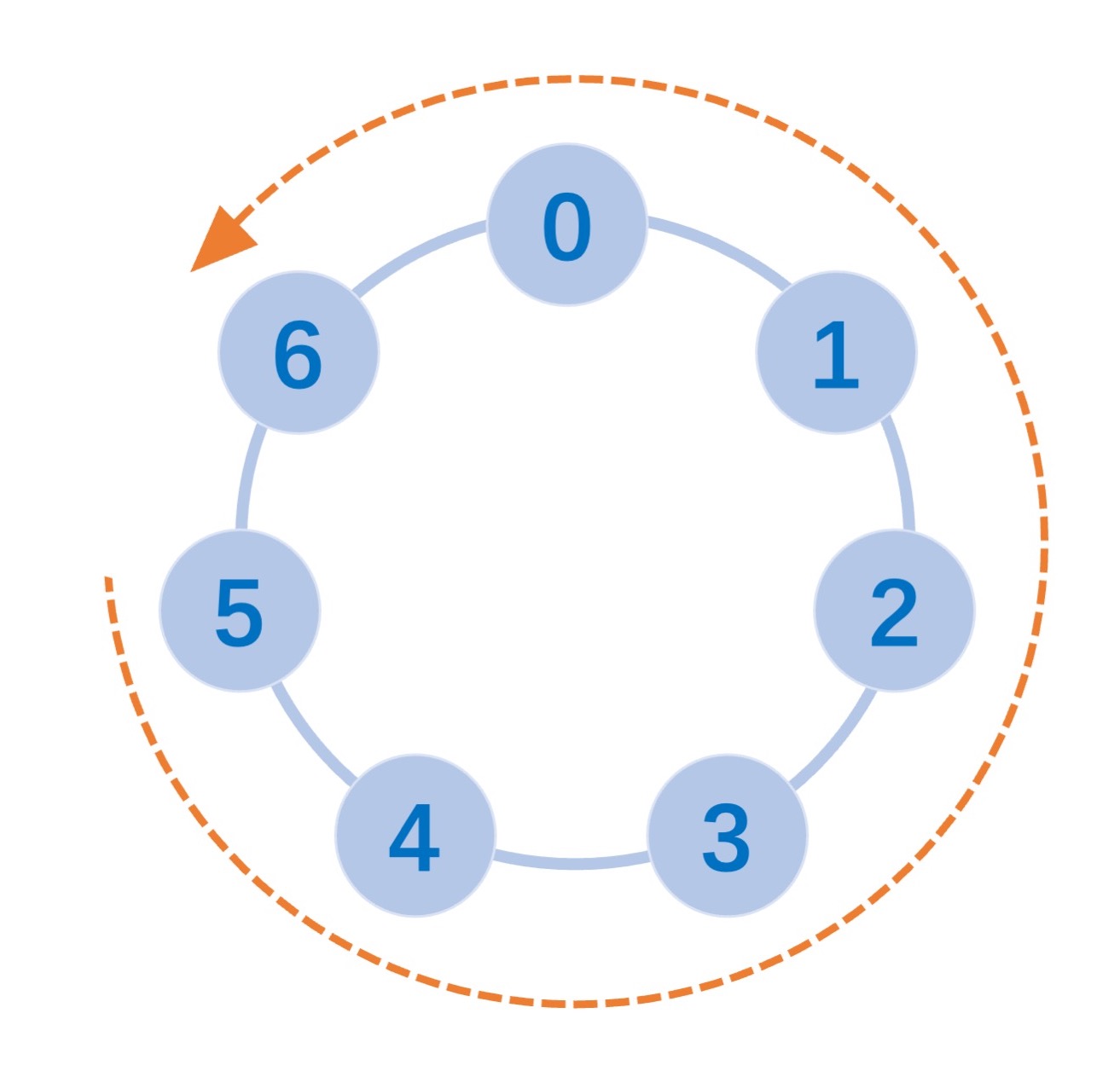}}}
	\caption{Interpolated tree, sequences $\Pi_i$'s, and a bad $\Pi_5$.\label{fig: 1tree}}
\end{figure}

It is immediate to verify that all $\Pi_i$ and $\Pi$ are continuous, in that any integer $j$ is sided by either $j+1$ or $j-1$. The key to the lower bound proof is to \emph{examine the continuous behaviors of $\Pi$ and $\Pi_i$'s, which are extracted from the tree, on the cycle}.

The shortest way to go continuously from $i$ to $i+1$ on the cycle is to directly use the edge $(i,i+1)$.
We say that $\Pi_i$ is \emph{good}\footnote{There is a notion in combinatorial topology called ``degree of a mapping'', and our notion ``good'' is similar to their ``degree-$0$'' case. For the purpose of technical overview, we do not go deep in formally presenting the notions.}, if the number of $(i,i+1)$ appearances in $\Pi_i$ exceeds that of $(i+1,i)$ by $1$.
For example, $\Pi_2=(2,{\color{red}3},{\color{red}4},5,{\color{red}4},3)$ is good, but $(2,{\color{red}1},{\color{red}0},{\color{red}n-1}\ldots,{\color{red}4},3)$ is not.
Intuitively, $\Pi_i$ is good if it ``essentially travels on the short side of cycle''.
See \Cref{fig: 1tree} for an illustration.
The following two properties are easy to prove:
\begin{itemize}
    \item if $\Pi_i$ is good, the number of $(j,j+1)$ appearances equals that of $(j+1,j)$ for each $j\ne i$; and
    \item if $\Pi_i$ is bad, it contains all $n$ numbers and therefore has length at least $n-1$.
\end{itemize}

The crucial observation is that, if some $\Pi_i$ is bad, then the distortion of $T$ is at least $n-1$, as the cycle-distance between $i,i+1$ is $1$ while the tree-distance between $i,i+1$ is $|\Pi_i|\ge n-1$. Therefore, it remains to show that not all $\Pi_i$'s are good. Assume for contradiction that they all are. Consider now the closed tour $\Pi$. From the first property of good $\Pi_i$'s, we know that for each $i$,
the $(i,i+1)$ appearances exceed $(i+1,i)$ by $1$. However, as $\Pi$ is a closed tour that goes continuously around the tree, every time it travels from $i$ to $i+1$ along an interpolating edge $(i,i+1)$, later it has to travel back along the same edge from $i+1$ to $i$, so the $(i,i+1)$ appearances should equal that of $(i+1,i)$, leading to a contradiction.

To summarize, we interpolate the tree edges, mapping the cycle into a continuous tree-tour $\Pi$ which consists of sub-trajectories $\Pi_0,\ldots,\Pi_{n-1}$, and examine their behaviors from the topological perspective. The key point is that a cycle must contain an edge $(i,i+1)$ whose corresponding trajectory $\Pi_i$ is bad, and such an edge, which we call a \emph{bad edge}, is the cause of $\Omega(n)$ distortion.

\subsection*{Case $k=2$: the torus grid}

The best distortion upper bound for the $k=2$ case is $\tilde O(\sqrt n)$  \cite{bartal2022covering}. Therefore, a natural first graph to try proving a matching lower bound on is the $\sqrt{n}\times \sqrt{n}$ grid. However, with a delicate design, two trees can actually cover the grid with distortion $O(1)$.
We present a sketch of the construction in \Cref{apd: two trees}. Two concurrent works \cite{le2025coveringeuclideanplanepair,bikeev2025treecoverssize2} independent to ours give a similar construction with a much detailed presentation.

According to the lesson from the $k=1$ case, to make it harder for trees to cover, we should plant more cycles on the grid. Specifically, we connect each first-row vertex to their corresponding last-row vertex, and connect each first-column vertex to their corresponding last-column vertex, so the resulting graph is a torus grid.  As we show next, this torus grid poses a $\Omega(\sqrt{n})$ distortion lower bound for two trees.

We focus on the following two types of non-contractible cycles: vertical/horizontal cycles, which are cycles that go around the grid in the vertical/horizontal way, respectively. See \Cref{fig: 2tree} for an illustration. Both types of cycles are of length at least $\sqrt{n}$.

\begin{figure}[h]
	\centering
	\subfigure[The torus grid: the down-edge from $f'$ goes to $a$ of the first row, and similar for vertices $b,\ldots,g,b'\ldots,f'$. A vertical cycle is shown in purple.]
	{\scalebox{0.084}{\includegraphics{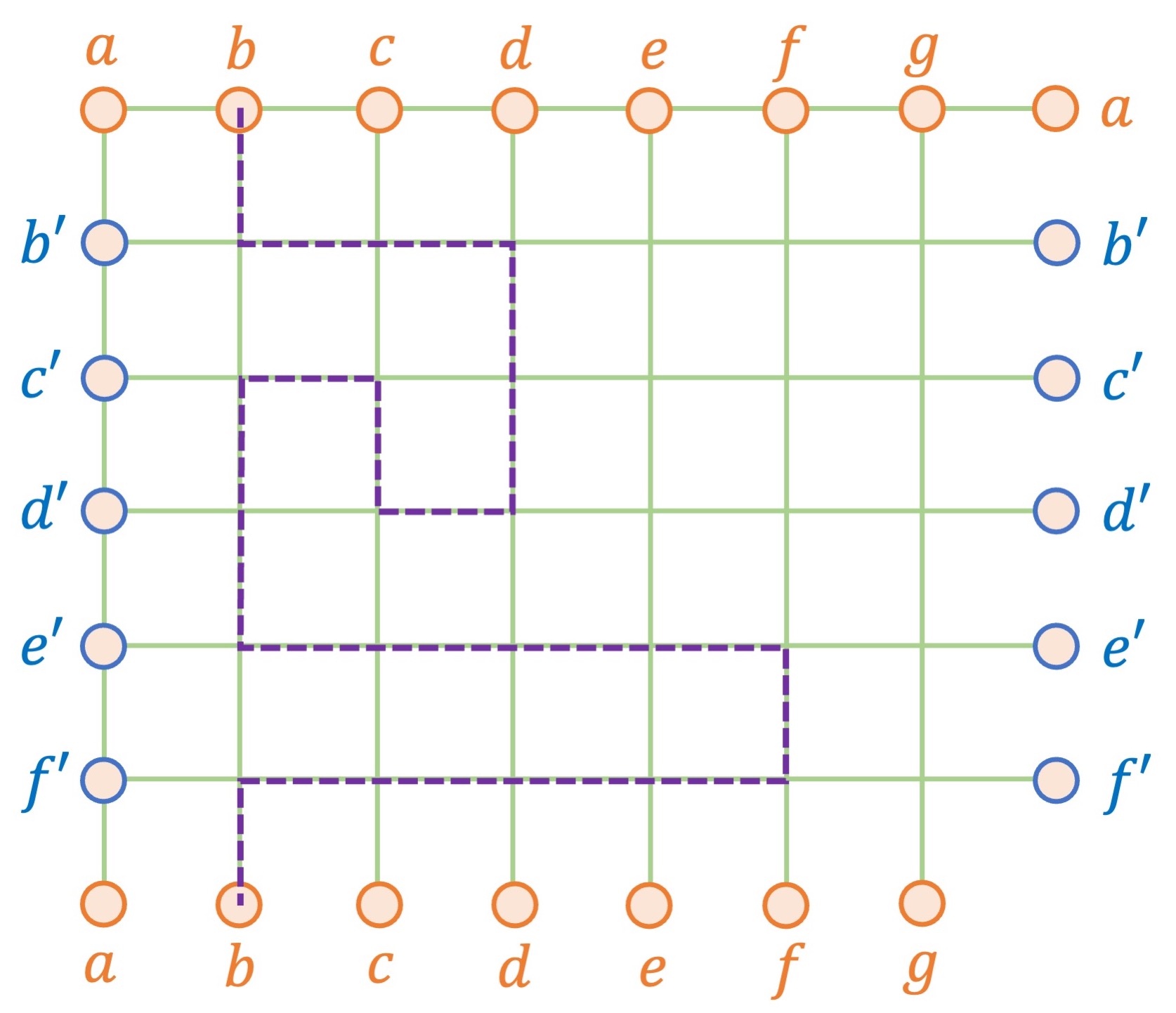}}}
	\hspace{0.2cm}
	\subfigure[Bad edges for $T_1$ are shown in black, forming a ``horizontal cut'': every vertical cycle has to contain at least one such edge.]
	{
		\scalebox{0.231}{\includegraphics{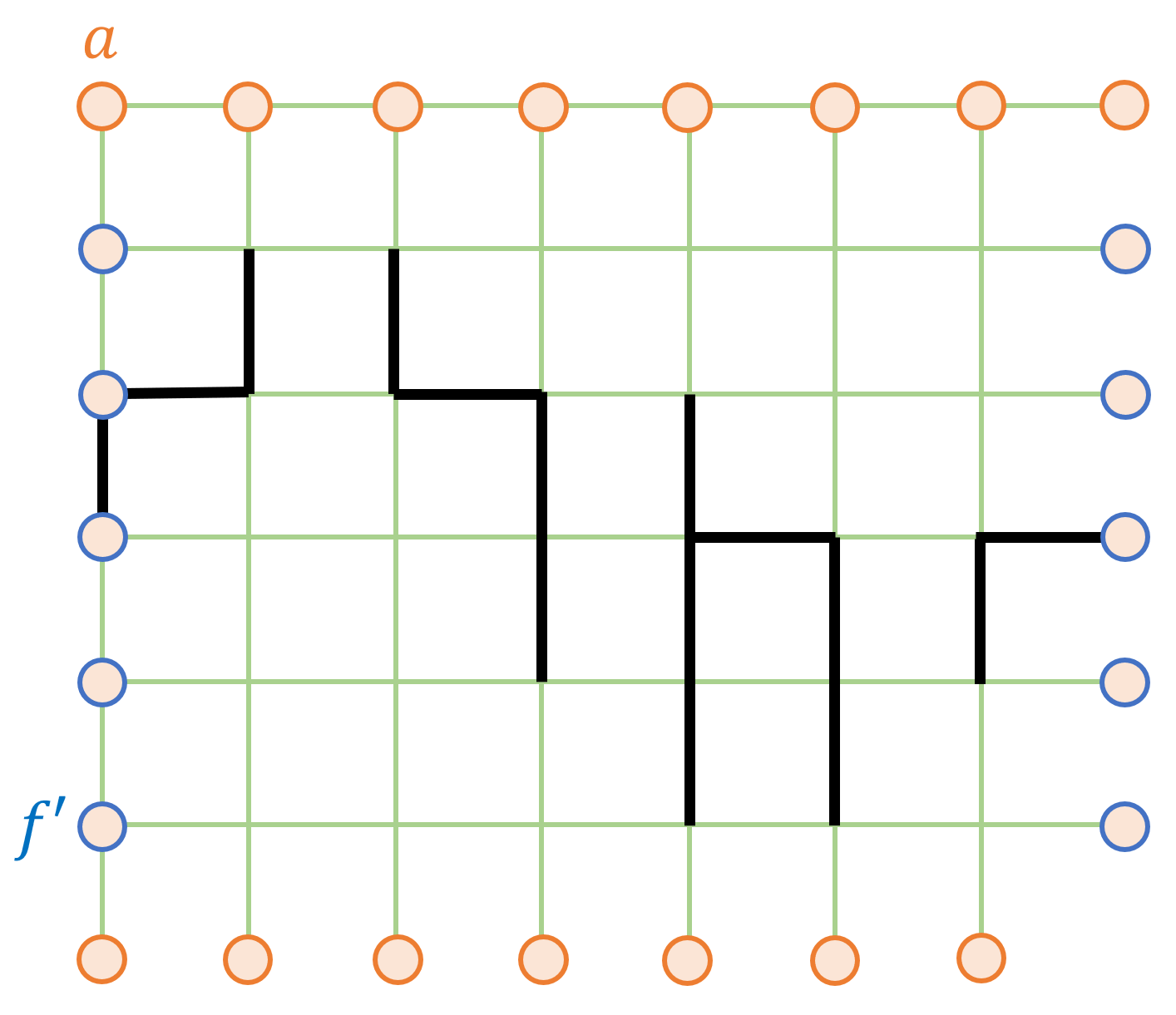}}}
        \hspace{0.2cm}
	\subfigure[A $T_2$-bad edge (red) shares an endpoint with a $T_1$-bad edge (black).]
	{
		\scalebox{0.091}{\includegraphics{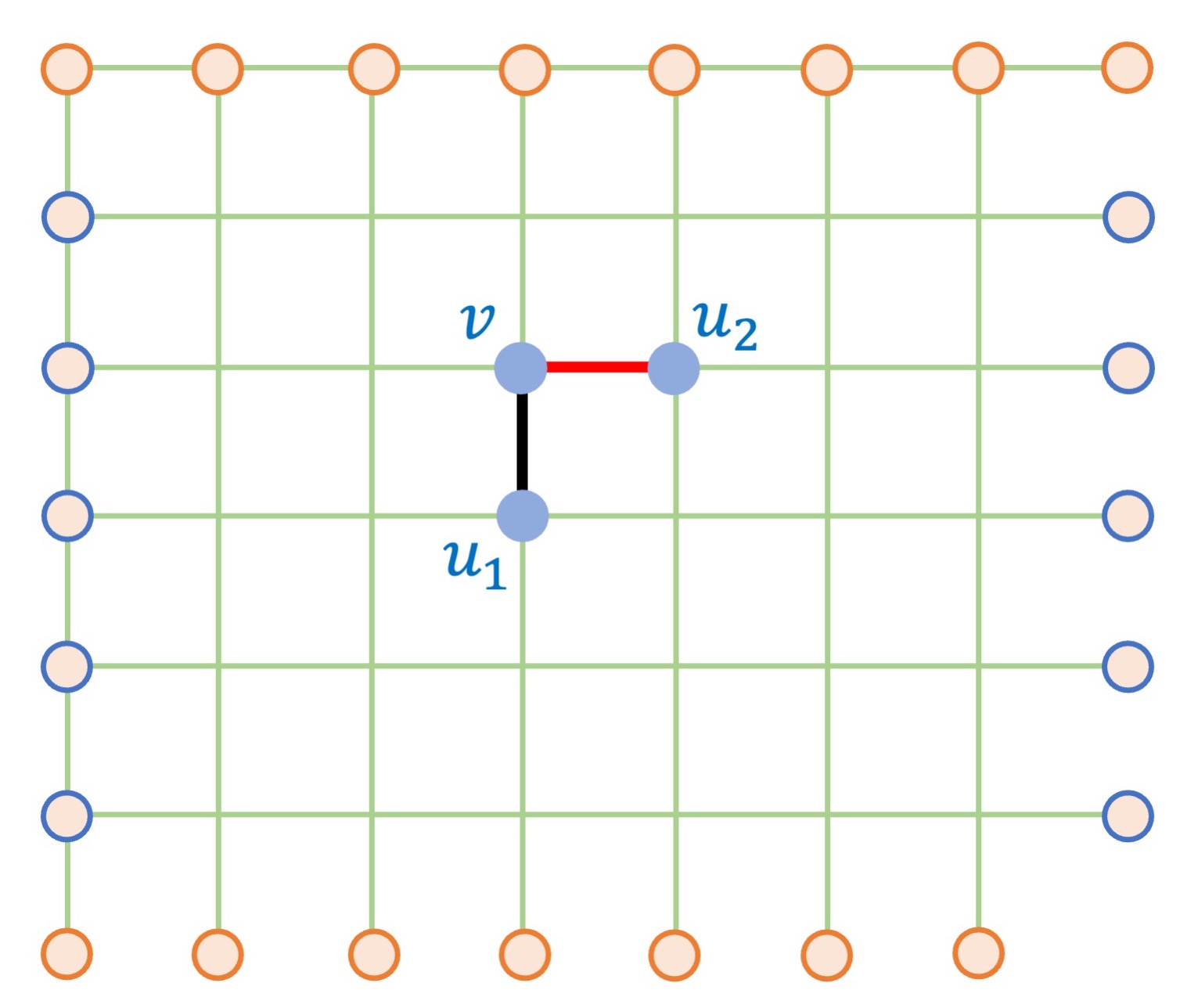}}}
	\caption{The torus grid, horizontal and vertical cycles, and  $T_1$ and $T_2$-bad edges.\label{fig: 2tree}}
\end{figure}

We denote the trees by $T_1,T_2$. To prove the $\Omega(\sqrt{n})$ distortion lower bound, it suffices to find a pair of vertices that are $O(1)$-close in the torus grid but $\Omega(\sqrt{n})$-far in both trees. The focus is still finding bad edges, similar to the $k=1$ case. At a high level, the lower bound proof can be summarized as follows.

\begin{enumerate}
\item Intuitively, every vertical cycle contains a bad edge for $T_1$. Therefore, in order to hit all vertical cycles, the union of all such edges must form a ``horizontal cut'' (see \Cref{fig: 2tree}).
\item Similarly, every horizontal cycle contains a bad edge for $T_2$. Therefore, in order to hits all horizontal cycles, the union of all such edges must form a ``vertical cut''.
\item The horizontal cut has the shape of a horizontal cycle, and so it must intersect the vertical cut. At their intersection we can find the desired pair of vertices. Specifically, as illustrated in \Cref{fig: 2tree}, a $T_1$-bad edge $(v,u_1)$ and a $T_2$-bad edge $(v,u_2)$ share a common endpoint $v$, so $\dist_{T_1}(v,u_1)=\Omega(\sqrt n)$ and $\dist_{T_2}(v,u_2)=\Omega(\sqrt n)$.
Then 
\begin{itemize}
    \item if $\dist_{T_2}(v,u_1)=\Omega(\sqrt n)$, then $v,u_1$ are $\Omega(\sqrt n)$-far  in both trees;
    \item if $\dist_{T_1}(v,u_2)=\Omega(\sqrt n)$, then $v,u_2$ are $\Omega(\sqrt n)$-far  in both trees;
    \item if $\dist_{T_2}(v,u_1),\dist_{T_1}(v,u_2)=o(\sqrt n)$, then by triangle inequality, 
    \[\dist_{T_1}(u_1,u_2)\ge \dist_{T_1}(v,u_1)-\dist_{T_1}(v,u_2)=\Omega(\sqrt n),\] 
    and similarly
    $\dist_{T_2}(u_1,u_2)=\Omega(\sqrt n)$, and so
    $u_1,u_2$ are $\Omega(\sqrt n)$-far in both trees.
\end{itemize}
\end{enumerate}

In fact, instead of the torus grid, our proof of \Cref{thm: main} (presented in \Cref{sec: tree cover}) actually works with the grid structure on the $2$-dimensional projective plane: that is, the graph obtained from the $\sqrt{n}\times \sqrt{n}$ grid by connecting (i) for each $i$, the $i$-th vertex of the first row to the $(\sqrt{n}-i)$-th vertex of the last row (rather than the $i$-th vertex of the last row as in the grid); and (ii) for each $i$, the $i$-th vertex of the first column to the $(\sqrt{n}-i)$-th vertex of the last column (rather than the $i$-th vertex of the last column as in the grid). The horizontal/vertical cycles and cuts can be defined analogously and analyzed similarly.

\subsection*{General $k$: the role of Tucker's Lemma}

Based on our approach in the $k=2$ case, we can set up a similar plan for proving lower bounds for general $k$: construct a $k$-dimensional grid with length $n^{1/k}$, make it cyclic by connecting $i$-dimensional antipodal points for each dimension $i$, find the ``dimension-$i$ cut'' consisting of $T_i$-bad edges, which is essentially a $(k-1)$-dimensional plane, and finally analyze the intersection of all $k$ planes.

The issue resides in the very last step (Step 3 in the $k=2$ case): we are only guaranteed that, around the intersection, for each $i$, some pair of points cross the dimension-$i$ cut, but that does not guarantee the existence of a pair of points simultaneously crossing all dimension-$i$ cuts.
Specifically, we consider the $k$-dimensional cube that represents the intersection of all dimension-cut planes. We associate with each vertex a vector in $\set{+1,-1}^k$ as follows. For each vertex, for each $i$, if it lies on the left side (left and right sides can be dictated in an arbitrary way) of the dimension-$i$ cut, then we set the $i$-th coordinate of its associated vector $v$ as $v_i=+1$, otherwise we set $v_i=-1$. Then, among all vectors in the cube,
\begin{itemize}
    \item what we can prove (via the results in \cite{musin2014sperner}): for each $i$, there exists a pair of vertices whose associated vectors $u,v$ satisfy that $u_i=-v_i$ (such a cube is called a \emph{centrally labeled} cube); and
    \item the goal is to show: there is a pair of vertices whose associated vectors $u,v$ satisfy that $u_i=-v_i$ for all $i$ (such a pair of vertices are said to form a \emph{complementary edge}).
\end{itemize} 

In fact, what we can prove does not necessarily imply the goal for $k\ge 3$. For example, the $8$ vectors in a $3$-dimensional cube can be 
\[\set{(1,1,1)\times 5, (1,-1,-1),(-1,-1,1),(-1,1,-1)}.\]
So the cube is indeed a centrally labeled cube, but there are no complementary edges. 

Luckily, the existence of a centrally labeled cube is enough to derive a $\Omega(n^{1/k})$ lower bound for Ramsey tree covers with $k$ trees. We present the details of this proof in \Cref{sec: ramsey}.

But to prove lower bound for general tree covers, we have to find complementary edges, and this is where we need Tucker's lemma, which states that, in a triangulated domain of $\mathbb{R}^d$, in a labeling of its vertices by integers in $\set{-d\ldots,-1,1,\ldots,d}$ that is antipodally symmetric on the domain's boundary, there is a complementary edge. Since the number of possible labels in our setting is $2^k$ (the size of set $\set{+1,-1}^k$), we have to significantly increase the number of dimensions in our construction from $k$ to $d=2^{k-1}$, which reduces the length of our grid-like structure from $n^{\frac{1}{k}}$ to $n^{\frac{1}{2^{k-1}}}$. Therefore, we are only able to prove an $\Omega(n^{\frac{1}{2^{k-1}}})$ stretch lower bound for general tree covers with this approach. An additional modification to our grid-like construction is needed in order to achieve the antipodally symmetric property required in Tucker's lemma: for the $i$-th dimension, instead of connecting pairs of vertices whose connection is parallel to the dimension-$i$ axis, we need to connect pairs of vertices that are antipodally symmetric. In light of this, instead of working with a cube-shape grid, we will work with an octahedron-shape grid, as it simplifies the analysis.
We present the details in \Cref{sec: tree cover}.

\section{Preliminaries: Simplices, Triangulations, and Tucker's Lemma}
A set $\sigma \subseteq \mathbb{R}^d$ is a \emph{$k$-dimensional simplex} (or \emph{$k$-simplex}) if it is the convex hull $\conv(A)$ of a set $A$ of $k+1$ affinely independent points. \emph{Vertices} of $\sigma$, denoted as $V(\sigma)$, refer to points in $A$.  Simplices of the form $\conv(B)$ where $B \subseteq A$ are called \emph{faces} of $\sigma$.

Our proof is based on some combinatorial fixed-point theorems. To formally state them, we need the notion of triangulation and its antipodal symmetry.
\begin{definition}[Triangulation]
Let $X$ be a subset of $\mathbb{R}^d$ that is homeomorphic to the $\ell_2$-norm unit ball. A finite family $\Gamma$ of simplices is a \emph{triangulation of} $X$, iff $\bigcup_{\sigma \in \Gamma}\sigma = X$, and the following two properties hold:
\begin{itemize}
\item For each simplex $\sigma \in \Gamma$, all faces of $\sigma$ are also simplices in $\Gamma$.
\item For any two simplices $\sigma_1,\sigma_2 \in \Gamma$, the intersection $\sigma_1 \cap \sigma_2$ is a face of both $\sigma_1$ and $\sigma_2$.
\end{itemize}
We define the vertices of $\Gamma$ as $V(\Gamma)=\bigcup_{\sigma\in \Gamma}V(\sigma)$.
We define the \emph{boundary} of $\Gamma$ to be the triangulation $\partial(\Gamma)\subseteq \Gamma$ containing all simplices in $\Gamma$ that lie entirely in the boundary $\partial(X)$ of $X$.
\end{definition}

\begin{definition}[Antipodal Symmetry]
A triangulation $\Gamma$ is said to be \emph{antipodally symmetric on the boundary} if, for all simplices $\sigma \in \partial(\Gamma)$, the opposite simplex $-\sigma$ ($-\sigma=\set{-v\mid v\in \sigma}$) also lies in $\partial(\Gamma)$.
\end{definition}


\begin{definition}[Labeling of Triangulation]
A labeling of triangulation $\Gamma$ is an assignment $\ell:V(\Gamma)\rightarrow L$ that assigns an integer of a set $L\subseteq \mathbb{Z}$ to each vertex of $\Gamma$.
\end{definition}

With the above definitions, we can now state Tucker's lemma (originally in \cite{tucker1945some}; see \cite{meunier2006pleins} for a survey) and its variant.
\begin{theorem}[Tucker's Lemma]\label{thm: tucker}
Let $\Gamma$ be a triangulation of the domain $X\subset\mathbb{R}^d$ that is antipodally symmetric on the boundary.  Let $\ell:V(\Gamma)\rightarrow \{-d,\ldots,-2,-1,1,2,\ldots,d\}$ be a labeling such that $\ell(v) = -\ell(-v)$ for all $v \in \partial(X)$.  Then $\Gamma$ contains a complementary edge, that is, an edge $(v,v')$ with $\ell(v)=-\ell(v')$.
\end{theorem}

We can further use $2^d$ labels, instead of $2d$ labels in Tucker's lemma, to obtain the following extension of Tucker's lemma (Theorem 1.9 in \cite{grant2013geometric}):

\begin{theorem}[\cite{grant2013geometric}]\label{thm: tuckercube}
Let $\Gamma$ be a triangulation of $X\subset \mathbb{R}^d$ that is antipodally symmetric on the boundary of the domain $X$.  Let $\ell:V(\Gamma)\rightarrow \set{+1,-1}^k$ be a labeling such that $\ell(v) = -\ell(-v)$ for all $v \in \partial(X)$.  Then $\Gamma$ contains a neutral simplex, that is, a simplex $\sigma$ such that for all $1\le i\le n$, there exist $v,v'\in V(\sigma)$ such that $\ell_i(v)=-\ell_i(v')$, where $\ell_i(v)$ is the $i^{\mathrm{th}}$ coordinate of $\ell(v)$.
\end{theorem}

\Cref{thm: tuckercube} can be proved by either applying Borsuk-Ulam Theorem \cite{matouvsek2003using} (a topology approach) or by extending the framework of Ky Fan \cite{fan1952generalization} to more possible labelings \cite{grant2013geometric} (a combinatorial approach).
We refer the interested readers to \cite{grant2013geometric} for a survey of relevant approaches.

\section{Proof of \Cref{thm: main}}
\label{sec: tree cover}

In this section, we provide the proof of \Cref{thm: main}.
We will construct a metric $(X,d_X)$, and show that for any $k$ dominating trees $T_1,\ldots,T_k$ on $X$, there is a pair of vertices $u,v$ in $X$, such that 
\[\min_{1\le i\le k}\set{\dist_{T_i}(u,v)}=\Omega\bigg(\frac{n^{\frac{1}{2^{k-1}}}}{2^{k-1}}\bigg)\cdot d_X(u,v).\]

\subsection{Construction of the Hard Instance}



We first define an edge-weighted graph $\tilde{G}$ as follows.
    \begin{itemize}
        \item Vertex set: $V(\tilde{G})=\{x\in\mathbb{Z}^d:||x||_1\leq r\}$;
        \item Edges set: $E(\tilde{G})=E_0\cup \Esp$, where edges in $\Esp$ are called \emph{special edges}, and
        \begin{itemize}
            \item $E_0=\{(u,v):||u-v||_1=1\}$,
            \item $\Esp=\{(u,v):u=-v,||u||_1=r\}$;
        \end{itemize}
        
        \item Edge weights: $w_{e}=
        \begin{cases}
            1, \textnormal{ if } e\in E_0;\\
            0, \textnormal{ if } e\in \Esp.\\
        \end{cases}$
    \end{itemize}
Our hard instance $G$ is the graph obtained by contracting all the edges in $\Esp$ within $\tilde{G}$. Specifically, we denote the two vertices contracted by the one with positive value at the first non-zero coordinate.
By choosing $d=2^{k-1}$ and $r=\Theta(n^{1/d})$, we can ensure that $|V(G)|=O(n)$.
The metric $(X,d_X)$ is simply the shortest path distance metric on all vertices of $G$.



\subsection{Reduction to Tucker's Lemma}

\begin{lemma}
There is a triangulation $\Gamma'$ of $\{x\in\mathbb{R}^d:||x||_1\leq r\}$, such that $V(\Gamma')=V(\tilde{G})$, and for every edge $(u,v)$ in $\Gamma'$, $||u-v||_{\infty}\leq1$.
\end{lemma}

\begin{proof}

Let $\set{e_1,\ldots,e_d}$ be the standard basis vectors in $\mathbb{R}^d$.
Given a simplex $\sigma$, we denote by $\mathcal{C}(\sigma)$ the set of all faces of $\sigma$, i.e., $\mathcal{C}(\sigma)=\bigcup_{A\subseteq V(\sigma)}\conv(A)$. For each vertex $v=\sum_{i=1}^{d}\alpha_ie_i\in\{x\in\mathbb{R}^d:||x||_1\leq 1\}$ where $\alpha_i\ne0$ for all $i$, we define 
$\sigma(v)=\conv(\{\mathrm{sign}(\alpha_i)\cdot e_i:i=1,\ldots,d\})$ (see \Cref{fig: l1normball_triangulation}).
Specially, we define $\sigma(v)=\emptyset$ if $\alpha_i=0$ for some $i$.
We further define $\Gamma=\bigcup_{||v||\leq 1}\mathcal{C}(\sigma(v))$, as illustrated in \Cref{fig: l1normball_triangulation}.

\begin{figure}[h]
	\centering
	\subfigure[A $2$-dimensional $\ell_1$-ball can be naturally partitioned into unit  balls.]
	{\scalebox{0.126}{\includegraphics{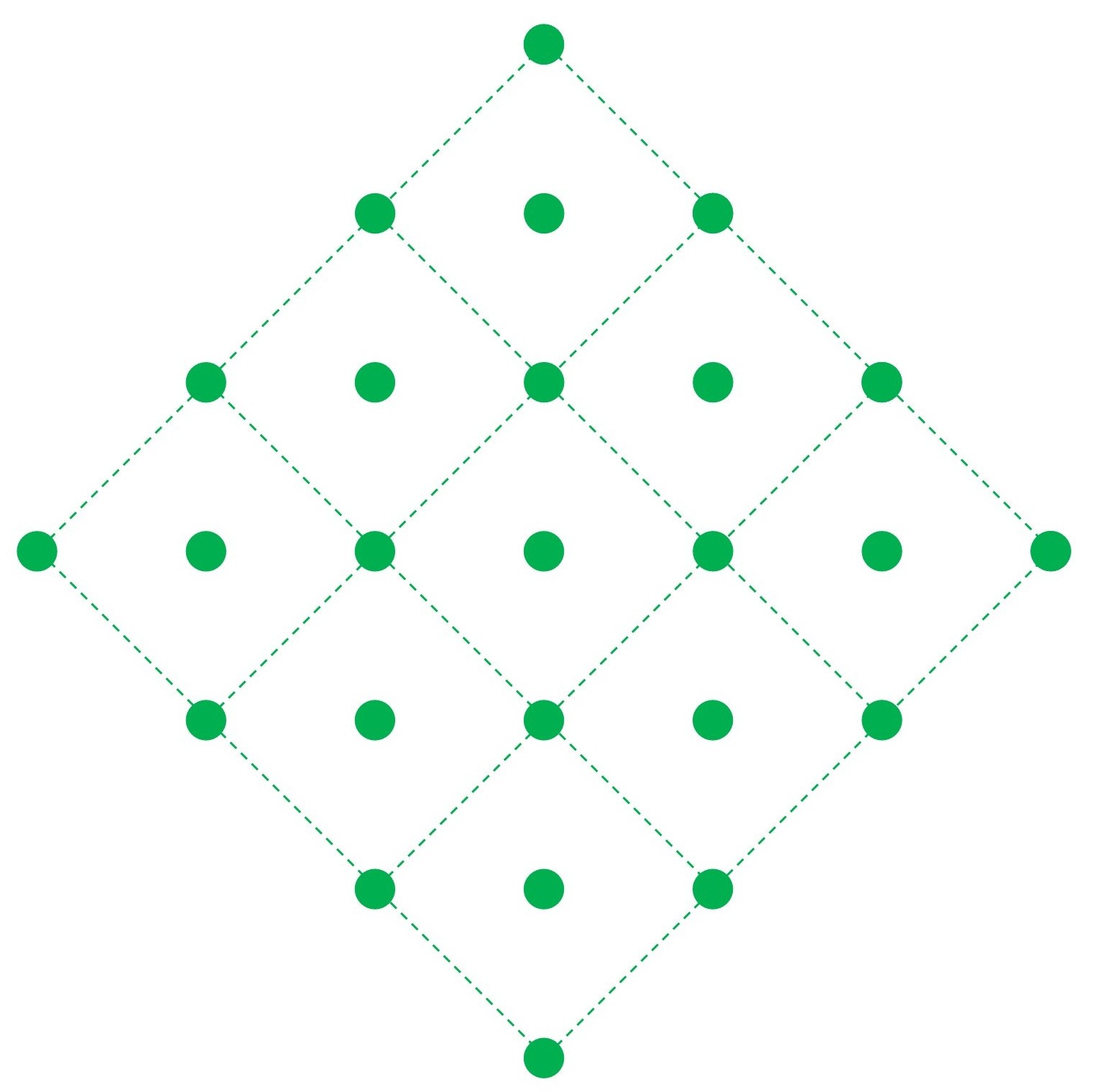}}}
	\hspace{0.2cm}
	\subfigure[Triangulation of the $2$-dimensional $\ell_1$-norm unit ball. $v$ lies in $\sigma(v)$, the simplex with red vertices.]
	{
		\scalebox{0.31}{\includegraphics{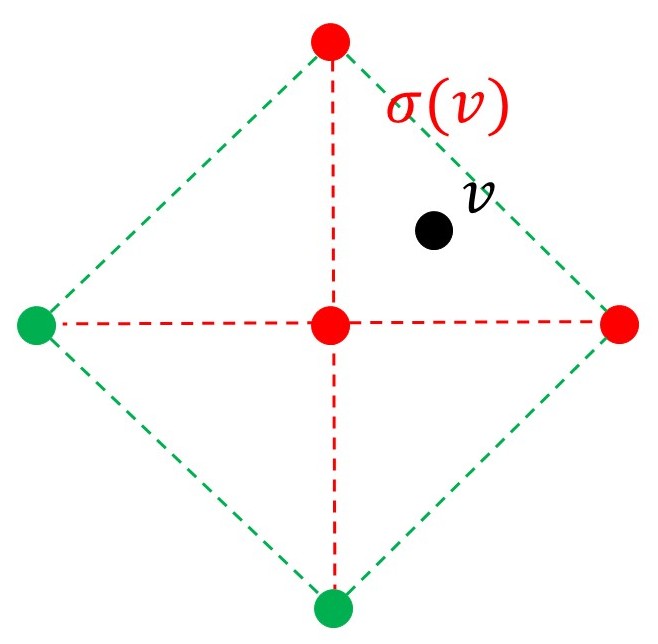}}}
        \hspace{0.2cm}
        \subfigure[Triangulation of the $3$-dimensional $\ell_1$-norm unit ball. $v$ lies in $\sigma(v)$, the simplex with red vertices.]
	{
		\scalebox{0.3}{\includegraphics{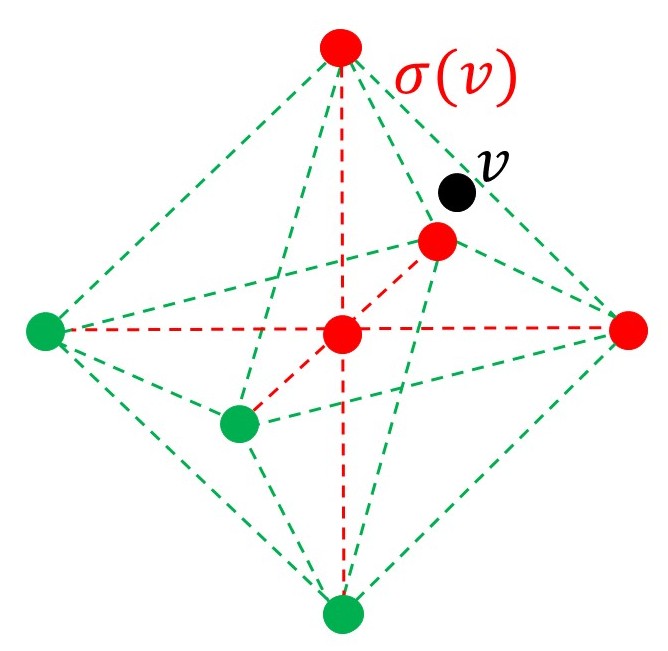}}}
	\caption{The $\ell_1$ norm ball grid and its decomposition, and the triangulations of $\ell_1$-norm unit balls in $2$ and $3$ dimensions.\label{fig: l1normball_triangulation}}
\end{figure}

For every vertex $v\in\{x\in\mathbb{R}^d:||x||_1\leq 1\}$, if $\alpha_i\ne0$ for all $i$, by $v\in\sigma(v)$, we have $v\in\bigcup_{\sigma\in\Gamma}\sigma$. Otherwise, $v\in\conv(\{\mathrm{sign}(\alpha_i)\cdot e_i:i=1,\ldots,d\})\in\Gamma$, where $\mathrm{sign}(x)=\begin{cases}
        1,&\textnormal{if } x\geq0;\\
        -1,&\textnormal{if } x<0.
    \end{cases}$. Thus $v\in\bigcup_{\sigma\in\Gamma}\sigma$ for every $v\in\{x\in\mathbb{R}^d:||x||_1\leq 1\}$. It follows that $\bigcup_{\sigma\in\Gamma}\sigma=\{x\in\mathbb{R}^d:||x||_1\leq 1\}$.
By the definition of $\Gamma$, for each simplex $\sigma \in \Gamma$, all the faces of $\sigma$ are also simplices in $\Gamma$. 
Furthermore, for every two simplices $\sigma_1,\sigma_2\in\Gamma$, $\sigma_1\cap\sigma_2=\conv(V(\sigma_1)\cap V(\sigma_2))\in\mathcal{C}(\sigma_1)\cap\mathcal{C}(\sigma_2)$.
This concludes that $\Gamma$ is a triangulation of $\{x\in\mathbb{R}^d:||x||_1\leq 1\}$.

$\{x\in\mathbb{R}^d:||x||_1\leq r\}$ can be naturally partitioned into $\ell_1$-norm unit balls, each of which can be triangulated in the same manner as $\Gamma$, as illustrated in \Cref{fig: l1normball_triangulation}. It gives a triangulation of $\{x\in\mathbb{R}^d:||x||_1\leq r\}$, denoted by $\Gamma'$. Note that $V(\Gamma')=V(\tilde{G})$, and for every edge $(u,v)$ in $\Gamma'$, $||u-v||_{\infty}\leq1$.
\end{proof}

Given a dominating tree $T$ of $G$. As in \Cref{sec: case k=1}, we may assume without loss of generality that for each edge $(u,v)$ in $T$, its length is the distance between  $u$ and $v$ in  $G$, namely $\dist_T(u,v)=\dist_G(u,v)$.
For each edge $(u,v)$ in $T$, we subdivide it by $\dist_T(u,v)-1$ \emph{interpolating nodes}, which are then labeled by consecutive vertices in the shortest path between $u,v$ in $G$ (arbitrarily chosen if there are more than one shortest path). Edges connecting the neighboring interpolating nodes are called \emph{interpolating edges}.

We define $\Pi(u,v)$ as the sequence of all the above nodes from $u$ to $v$. Note that an interpolating edge may not be an edge in $\tilde{G}$. For each such edge $(x,y)$, observe that either $||x||_1=r$ or $||y||_1=r$. We may assume without loss of generality that $||x||_1=r$. One can observe that $(-x,y)$ is included in the edge set of $\tilde{G}$. Thus we subdivide $(x,y)$ by $(x,-x)$ and $(-x,y)$. By the above subdivision process, we obtain a new sequence consisting of edges in $\tilde
{G}$, and we denote it by $\tilde{\Pi}(u,v)$.

For vertices $v,v'\in G$, let $\big(v_0(=v),v_1,\ldots,v_{\ell}(=v')\big)$ be the $v$-$v'$ path in $T$. We define $P_{\tilde{G}}^T(v,v')$ as the sequential concatenation of sequences $\tilde{\Pi}(v_0,v_1),\tilde{\Pi}(v_1,v_2),\ldots,\tilde{\Pi}(v_{\ell-1},v_{\ell})$, identifying, for each $i$, the last node of $\tilde{\Pi}(v_{i-1},v_{i})$ and the first node of $\tilde{\Pi}(v_{i},v_{i+1})$. 

Now for a vertex pair $v,v'\in \tilde{G}$, if $v,v'$ are also in $G$, the definition of $P_{\tilde{G}}^T(v,v')$ is the same as before. Otherwise, we may assume $v'$ is in $\tilde{G}\setminus G$, namely, $||v'||=r$ with negative value at the first non-zero coordinate. If $P_{\tilde{G}}^T(v,-v')$ is ended by $(v',-v')$, then we truncate it and obtain a sequence ended by $v'$. Otherwise, we concatenate it by $(-v',v')$. The new sequence is defined as $P_{\tilde{G}}^T(v,v')$.

Denote the length of a path $P$ as $\ell(P)$. Observe that the tree-distance between vertices $v,v'$ is $\ell(P_{\tilde{G}}^T(v,v'))$ for each vertex pair $v,v'$ in $G$.



\begin{definition}[Tree Index]
    Given a tree $T$ and a vertex $v\in V(\tilde{G})$. 
    Denote by $c$ the vertex corresponding to the vector $\vec{0}$.
    We define an index $i_T(v)$ of $v$ for $T$ as:
    \[i_T(v)=
    \begin{cases}
        1, & P_{\tilde{G}}^T(c,v) \text{ contains an odd number of special edges};\\
        -1, & P_{\tilde{G}}^T(c,v) \text{ contains an even number of special edges.}
    \end{cases}
    \]
For the trees $T_1,\ldots,T_k$ of the tree cover, we define the \emph{tree index} of $v$ as $i(v)=(i_{T_1}(v),\ldots,i_{T_k}(v))$.
\end{definition}


By definition, $i(v)=-i(-v)$ for all $||v||_1=r$. Note that tree indices on $V(\tilde{G})$ is a labeling $V(\Gamma')\to \set{+1,-1}^k$ of 
the triangulation $\Gamma'$.
We now show that neighboring vertices with opposite indices have large tree distances.

\begin{lemma}\label{lem: complemenntary edge implies far}
For a tree $T$ and vertices $u,v$ with $||u-v||_{\infty}\leq1$, if $i_{T}(u)= -i_{T}(v)$, then $\dist_{T}(u,v)\geq 2r-d$.
\end{lemma}

The proof of \Cref{lem: complemenntary edge implies far} requires the following claim.

\begin{claim}
\label{lem: odd times transportation far than straight line distance}
    Given a vertex pair $u,v\in V(\tilde{G})$ and a path $P$ between $u$ and $v$ in $\tilde{G}$. If $|E(P)\cap \Esp|$ is an even number, then $\ell(P)\geq||u-v||_1$.
\end{claim}
\begin{proof}
    \begin{figure}[h]
	\centering
        \scalebox{0.16}{
        \includegraphics{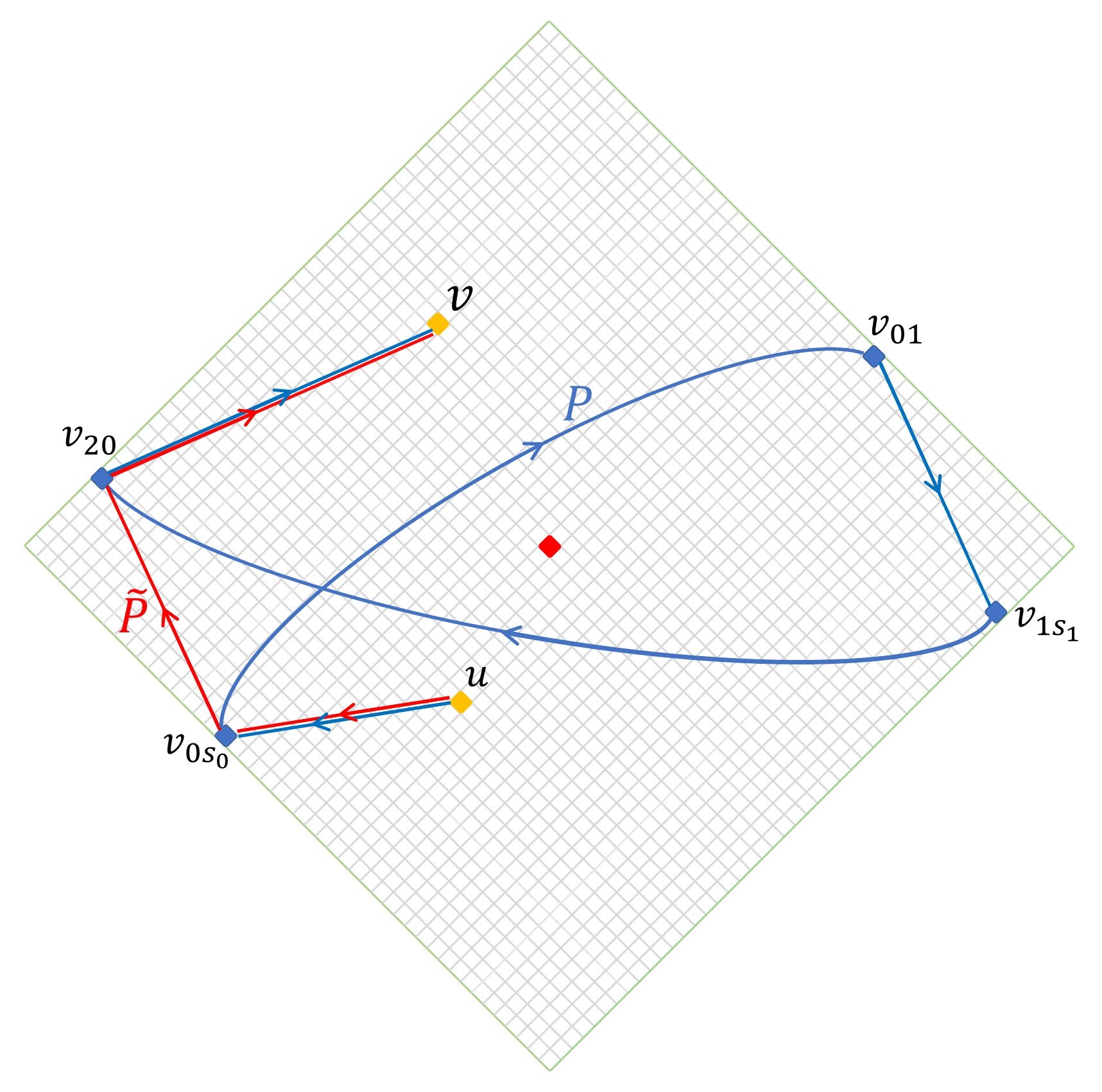}}
	\caption{Construction of $\tilde{P}$ by alternately flipping $P$. The two paths have the same length.}
\end{figure}
    Without loss of generality, we assume \[P=\big(v_{00}(=u),v_{0 1},\ldots, v_{0s_0},v_{10},\ldots, v_{1s_1},v_{20},\ldots, v_{(t-1)s_{t-1}},v_{t0},\ldots, v_{ts_t}(=v)\big),\] where for each $j=0,\ldots,t-1$, $v_{js_j}=-v_{(j+1)0}$ and $||v_{js_j}||_1=r$. Then we construct \[\tilde{P}=\big(v_{00},v_{0 1},\ldots, v_{0s_0},-v_{10},\ldots ,-v_{1s_1},v_{20},\ldots, (-1)^{t-1}v_{(t-1)s_{t-1}},(-1)^tv_{t0},\ldots, (-1)^tv_{ts_t}\big).\]

    It can be observed that $\tilde{P}$ is a path in $\tilde{G}\setminus\Esp$, namely containing no edges in $\Esp$. Since $t=|P\cap \Esp|$, $2\mid t$, and so $(-1)^tv_{ts_t}=v_{ts_t}=v$. Therefore, $\ell(P)=\ell(\tilde{P})\geq||u-(-1)^tv_{ts_t}||_1=||u-v||_1$.
\end{proof}

\begin{proof}[Proof of \Cref{lem: complemenntary edge implies far}]
Since $i_T(u)=-i_T(v)$, without loss of generality we can assume that $2\nmid |P_{\tilde{G}}^T(c,u)\cap\Esp|$ and $2\mid |P_{\tilde{G}}^T(c,v)\cap\Esp|$. We consider the tree $T$ rooted at $c$. If $u$ and $v$ are in different subtrees, $P_{\tilde{G}}^T(u,v)$ is the concatenation of $P_{\tilde{G}}^T(u,c)$ and $P_{\tilde{G}}^T(c,v)$. Otherwise, we may assume $u$ is an ancestor of $v$, then $P_{\tilde{G}}^T(c,u)$ is a subsequence of $P_{\tilde{G}}^T(c,v)$. In both cases, 
\[|P_{\tilde{G}}^T(u,v)\cap\Esp|\equiv|P_{\tilde{G}}^T(c,u)\cap\Esp|+|P_{\tilde{G}}^T(c,v)\cap\Esp|\mod2,\] and so $2\nmid |P_{\tilde{G}}^T(u,v)\cap\Esp|$.
    
We can assume without loss of generality that $P_{\tilde{G}}^T(u,v)$ consecutively pass through the vertex $w$ and $-w$ with $||w||_1=r$, such that the first part split from $P_{\tilde{G}}^T(u,v)$ (denoted by $P(u,w)$) contains no vertex pair $i_1,i_2$ with $i_1=-i_2$ and $||i_1||_1=r$. The other part split from $P_{\tilde{G}}^T(u,v)$ is denoted as $P(-w,v)$. It follows that $2\mid |P(u,w)\cap\Esp|$ and $2\mid |P(-w,v)\cap\Esp|$. By \Cref{lem: odd times transportation far than straight line distance}, we have $\ell(P(u,w))\geq||u-w||_1$ and $\ell(P(-w,v))\geq||-w-v||_1$. Therefore, 
\[
\begin{split}
\dist_{T}(u,v) & =\ell(P_{\tilde{G}}^{T}(u,v))=\ell(P(u,w))+\ell(P(-w,v))\\
& \geq||u-w||_1+||-w-v||_1\geq||w-(-w)||_1-||u-v||_1\geq2r-d\cdot||u-v||_{\infty}\geq 2r-d.
\end{split}\]
\end{proof}

Observe that \Cref{thm: tucker} requires only 
$d$ pairs of opposite labels, rather than $\{-d,\ldots,-2,-1,1,2,\ldots,d\}$. As we mentioned before, $i(v)=-i(-v)$ for all vertices $v$ with $||v||_1=r$. Thus, by \Cref{thm: tucker}, the triangulation $\Gamma'$ contains a complementary edge. That is, there exists a pair of neighboring vertices $u,v\in \Gamma'$ such that $i(u)=-i(v)$ and $||u-v||_{\infty}\leq 1$.

From \Cref{lem: complemenntary edge implies far}, for each tree $T_j$ in the tree cover $\set{T_1,\ldots,T_k}$, $\dist_{T_j}(u,v)\geq 2r-d$. Note that $\dist_{G}(u,v)\leq||u-v||_1\leq d\cdot||u-v||_{\infty}\leq d$. Therefore, 
\[\min_{1\le j\le k}\set{\dist_{T_j}(u,v)}=\Omega\bigg(\frac{r}{d}\bigg)\cdot \dist_{G}(u,v)=\Omega\bigg(\frac{n^{\frac{1}{2^{k-1}}}}{2^{k-1}}\bigg)\cdot \dist_{G}(u,v).\]

\section{Lower Bound for Ramsey Tree Covers}
\label{sec: ramsey}
Given a metric space $(X,d_X)$ on $|X|=n$ points, a \emph{Ramsey tree cover} is a collection $\tset=\set{T_1,\ldots,T_k}$ of (edge-weighted) trees on $X$, such that
\begin{itemize}
    \item for each $1\le i\le k$ and every pair $x,y$ of vertices in $X$, $\dist_{T_i}(x,y)\ge d_X(x,y)$ (that is, each $T_i$ is a \emph{dominating tree}); and
    \item for every vertex $x$ in $X$, there exists a tree $T_i\in \tset$ (called the \emph{home tree} of $x$), such that, for every vertex $y\in X$, $\dist_{T_i}(x,y)\le \alpha\cdot d_X(x,y)$. 
\end{itemize}
$k=|\tset|$ is called the \emph{size} of $\tset$ and $\alpha\ge 1$ is called the \emph{distortion} or \emph{stretch} of $\tset$.

In this section, we provide an alternative proof of the near-optimal bound of $\Omega_k(n^{1/k})$ for Ramsey tree covers of size $k$. We use a metric $(X,d_X)$ similar to the one introduced in \Cref{sec: tree cover}, and show that for every $k$ dominating trees $T_1,\ldots,T_k$ on $X$, there is a vertex $u$ in $X$, such that 
\[\min_{1\le i\le k}\set{\max_{v\in X\setminus\set{u}}\frac{\dist_{T_i}(u,v)}{d_X(u,v)}}=\Omega\bigg(\frac{n^{\frac{1}{k}}}{k}\bigg).\]

\subsection{Construction of the Hard Instance}
We first define an edge-weighted graph $\tilde{G}$ as follows.
    \begin{itemize}
        \item Vertex set: $V(\tilde{G})=\{x\in\mathbb{Z}^d:||x||_{\infty}\leq r\}$;
        \item Edges set: $E(\tilde{G})=E_0\cup \Esp$, where edges in $\Esp$ are called \emph{special edges}, and
        \begin{itemize}
            \item $E_0=\{(u,v):||u-v||_1=1\}$,
            \item $\Esp=\{(u,v):u=-v,||u||_{\infty}=r\}$;
        \end{itemize}
        
        \item Edge weights: $w_{e}=
        \begin{cases}
            1, \textnormal{ if } e\in E_0;\\
            0, \textnormal{ if } e\in \Esp.\\
        \end{cases}$
    \end{itemize}
Our hard instance $G$ is the graph obtained by contracting all the edges in $\Esp$ within $\tilde{G}$. Specifically, we denote the two vertices contracted by the one with positive value at the first non-zero coordinate.
By choosing $d=k$ and $r=\Theta(n^{1/d})$, we can ensure that $|V(G)|=O(n)$.
The metric $(X,d_X)$ is simply the shortest path distance metric on all vertices of $G$.


\subsection{Reduction to \Cref{thm: tuckercube}}

\begin{lemma}
There is a triangulation $\Gamma'$ of $\{x\in\mathbb{R}^d:||x||_{\infty}\leq r\}$, such that $V(\Gamma')=V(\tilde{G})$, and for every edge $(u,v)$ in $\Gamma'$, $||u-v||_{\infty}\leq1$.
\end{lemma}

\begin{proof}
For a simplex $\sigma$, we denote by $\mathcal{C}(\sigma)$ the set of all its faces, i.e., $\mathcal{C}(\sigma)=\bigcup_{A\subset V(\sigma)}\conv(A)$.
For a simplex $\sigma$ and a vertex $v$ affinely independent with $V(\sigma)$, denote $\mathrm{Ext}(\sigma,v)=\conv(V(\sigma)\cup\{v\})$ as the simplex obtained by extending $\sigma$ with $v$.

\begin{figure}[h]
	\centering
	\subfigure[A $2$-dimensional $\ell_{\infty}$-ball can be naturally partitioned into unit  balls.]
	{\scalebox{0.12}{\includegraphics{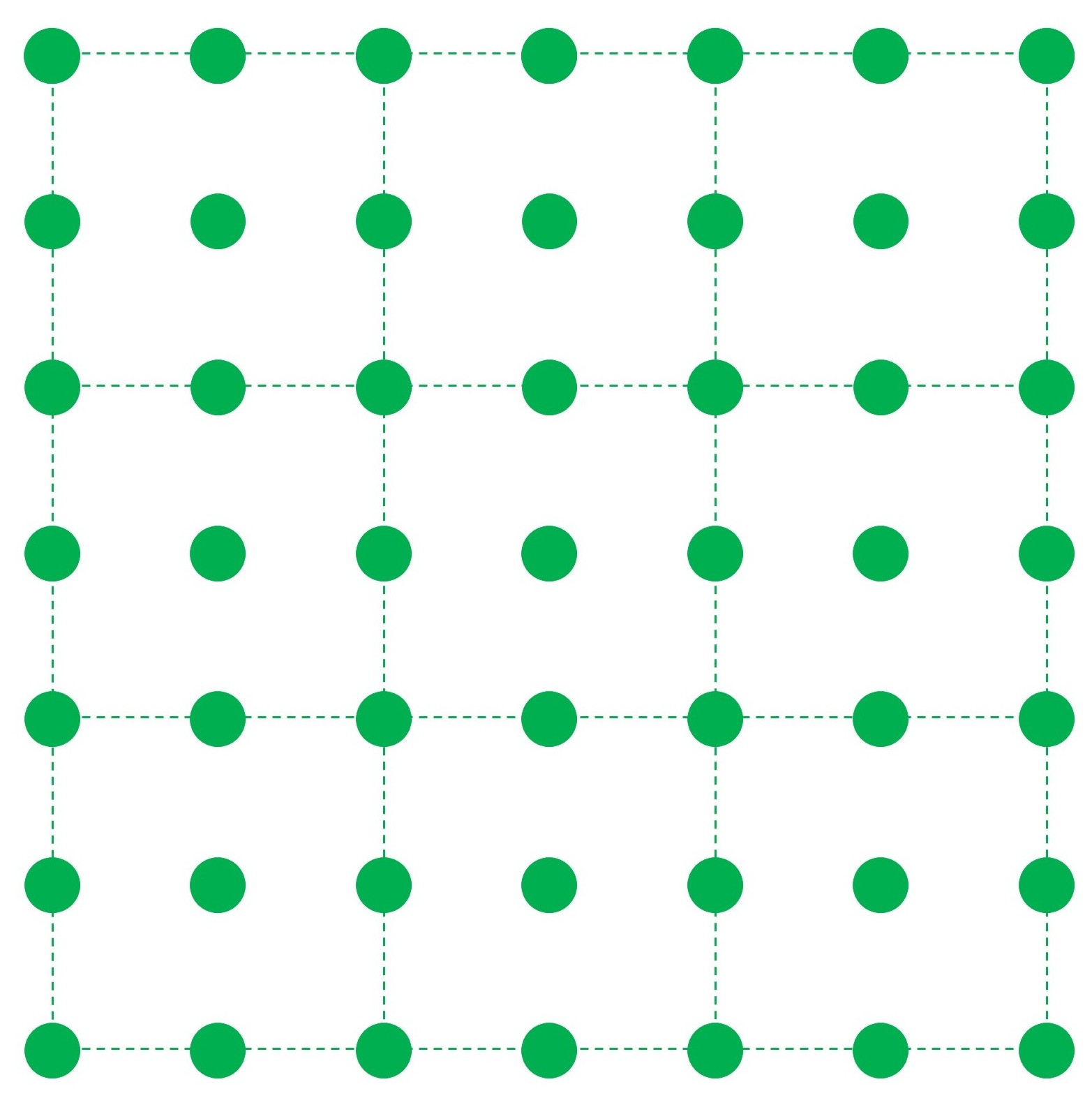}}}
	\hspace{0.8cm}
	\subfigure[Triangulation of the $2$-dimensional $\ell_{\infty}$-norm unit ball.]
	{
		\scalebox{0.32}{\includegraphics{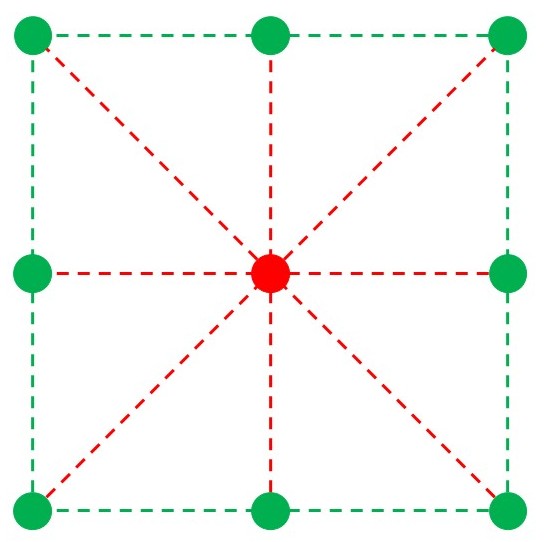}}}
        \hspace{0.8cm}
        \subfigure[Triangulation of $3$-dimensional $\ell_{\infty}$-norm unit ball. The red point is the center. Blue points form a $2$ dimensional $\ell_{\infty}$-norm unit ball.]
	{
		\scalebox{0.26}{\includegraphics{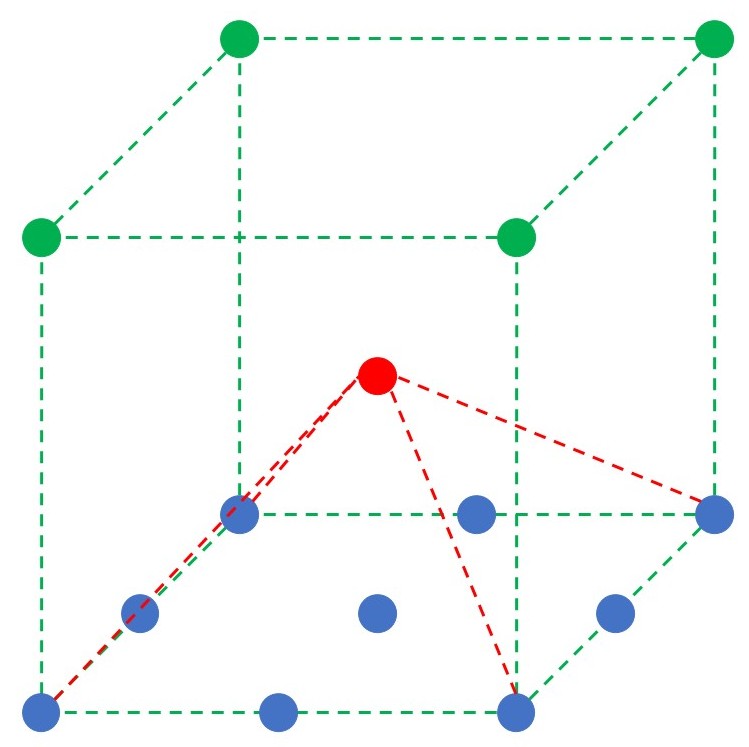}}}
	\caption{The $\ell_{\infty}$ norm ball grid and its decomposition, and the triangulations of $\ell_{\infty}$-norm unit balls in $2$ and $3$ dimensions.\label{fig: grid_triangulation}}
\end{figure}

We define $\Gamma$ via the following inductive construction:

For $d=2$, we define $\Gamma_0=\{(u,v)\in\mathbb{Z}^2:||u||_{\infty}=||v||_{\infty}=1,||u-v||_1=1\}$. Denote by $c$ the vertex corresponding to vector $\vec{0}$. Observe that the set $\{x\in\mathbb{R}^2:||x||_{\infty}\leq1\}$ can be triangulated by $\Gamma=\bigcup_{\sigma\in\Gamma_0}\mathcal{C}\big(\mathrm{Ext}(\sigma,c)\big)$. See \Cref{fig: grid_triangulation}.

For general $d>2$, we define \[F:=\bigg\{\{x=(x_1,\ldots,x_d)\in\mathbb{R}^d:||x||_{\infty}=1,x_j=\beta\}:j=1,\ldots,d,\beta=\pm1\bigg\},\]namely, the $d-1$ dimensional faces of the $\ell_{\infty}$-norm unit ball. By induction, for each face $f$ in $F$, there exists a triangulation of $f$. We denoted it as $\Gamma(f)$. We further define $\Gamma_{d-1}=\bigcup_{f\in F}\Gamma(f)$. Observe that $\{x\in\mathbb{R}^d:||x||_{\infty}\leq1\}$ can be triangulated by $\Gamma=\bigcup_{\sigma\in\Gamma_{d-1}}\mathcal{C}\big(\mathrm{Ext}(\sigma,c)\big)$. (See \Cref{fig: grid_triangulation})

By induction, one can verified that $\Gamma$ is a triangulation of $\{x\in\mathbb{R}^d:||x||_{\infty}\leq 1\}$.

The set $\{x\in\mathbb{R}^d:||x||_{\infty}\leq r\}$ can be naturally partitioned into $\ell_{\infty}$-norm unit balls, each of which can be triangulated in the same way as in $\Gamma$, as illustrated in \Cref{fig: grid_triangulation}. It gives a triangulation of $\{x\in\mathbb{R}^d:||x||_{\infty}\leq r\}$ denoted by $\Gamma'$. Note that $V(\Gamma')=V(\tilde{G})$, and, for every edge $(u,v)$ in $\Gamma'$, $||u-v||_{\infty}\leq1$.
\end{proof}

Given a dominating tree $T$ of $G$. For a pair $v,v'$ of vertices in $\tilde{G}$, we define $P_{\tilde{G}}^T(v,v')$ using the same definition as in \Cref{sec: tree cover}. For each vertex $v\in \tilde{G}$, we define the tree index of $v$ analogously. By definition, $i(v)=-i(-v)$ for all $||v||_{\infty}=r$. Note that tree indices on $V(\tilde{G})$ is a labeling $V(\Gamma')\to \set{+1,-1}^k$ of 
the triangulation $\Gamma'$.

By the following lemma similar to \Cref{lem: complemenntary edge implies far}, we show that neighboring vertices with opposite indices have large tree distances. The proofs are almost the same as the corresponding proofs (with small changes) in \Cref{lem: complemenntary edge implies far}. We present the proofs here for the sake of completeness. 

\begin{lemma}\label{lem: ramsey complemenntary edge implies far}
For a tree $T$ and vertices $u,v$ with $||u-v||_{\infty}\leq1$, if $i_{T}(u)= -i_{T}(v)$, then $\dist_{T}(u,v)\geq 2r-d$.
\end{lemma}
Similarly, the proof of \Cref{lem: ramsey complemenntary edge implies far} requires the following lemma.

\begin{claim}
\label{lem: ramsey odd times transportation far than straight line distance}
    Given a vertex pair $u,v\in V(\tilde{G})$ and a path $P$ between $u$ and $v$ in $\tilde{G}$. If $|E(P)\cap \Esp|$ is an even number, then $\ell(P)\geq||u-v||_1$.
\end{claim}
\begin{proof}
    Without loss of generality, we assume \[P=\big(v_{00}(=u),v_{0 1},\ldots, v_{0s_0},v_{10},\ldots, v_{1s_1},v_{20},\ldots, v_{(t-1)s_{t-1}},v_{t0},\ldots, v_{ts_t}(=v)\big),\] where for each $j=0,\ldots,t-1$, $v_{js_j}=-v_{(j+1)0}$ and $||v_{js_j}||_{\infty}=r$. Then we construct \[\tilde{P}=\big(v_{00},v_{0 1},\ldots, v_{0s_0},-v_{10},\ldots ,-v_{1s_1},v_{20},\ldots, (-1)^{t-1}v_{(t-1)s_{t-1}},(-1)^tv_{t0},\ldots, (-1)^tv_{ts_t}\big).\]

    It can be observed that $\tilde{P}$ is a path in $\tilde{G}\setminus\Esp$, namely containing no edges in $\Esp$. Since $t=|P\cap \Esp|$, $2\mid t$. Then $(-1)^tv_{ts_t}=v_{ts_t}=v$. Therefore, $\ell(P)=\ell(\tilde{P})\geq||u-(-1)^tv_{ts_t}||_1=||u-v||_1$.
\end{proof}

\begin{proof}[Proof of \Cref{lem: ramsey complemenntary edge implies far}]
Since $i_T(u)=-i_T(v)$, by the same analysis in the proof of \Cref{lem: complemenntary edge implies far}, $2\nmid |P_{\tilde{G}}^T(u,v)\cap\Esp|$.
We can assume without loss of generality that $P_{\tilde{G}}^T(u,v)$ consecutively pass through the vertex $w$ and $-w$ with $||w||_{\infty}=r$, such that the first part split from $P_{\tilde{G}}^T(u,v)$ (denoted by $P(u,w)$) contains no vertex pair $i_1,i_2$ with $i_1=-i_2$ and $||i_1||_{\infty}=r$. The other part split from $P_{\tilde{G}}^T(u,v)$ is denoted as $P(-w,v)$. It follows that $2\mid |P(u,w)\cap\Esp|$ and $2\mid |P(-w,v)\cap\Esp|$. By \Cref{lem: ramsey odd times transportation far than straight line distance}, we have $\ell(P(u,w))\geq||u-w||_1$ and $\ell(P(-w,v))\geq||-w-v||_1$. Therefore, 
\[
\begin{split}
\dist_{T}(u,v) & =\ell(P_{\tilde{G}}^{T}(u,v))=\ell(P(u,w))+\ell(P(-w,v))\\
& \geq||u-w||_1+||-w-v||_1\geq||w-(-w)||_1-||u-v||_1\geq2r-d\cdot||u-v||_{\infty}\geq 2r-d.
\end{split}
\]
\end{proof}

As we mentioned before, $i(v)=-i(-v)$ for all vertices $v$ with $||v||_{\infty}=r$. Thus, by \Cref{thm: tuckercube}, the triangulation $\Gamma'$ contains a neutral simplex. That is, there exists a simplex $\sigma$ such that for all $1\le j\le d$ ($d=k$), there are $v,v'\in V(\sigma)$ satisfying $i_{T_j}(v)=-i_{T_j}(v')$.

We arbitrarily choose a vertex $u$ in $V(\sigma)$. For each tree $T_j$ in the tree cover $\set{T_1,\ldots,T_k}$, there exist $v,v'\in V(\sigma)$ such that $i_{T_j}(v)=-i_{T_j}(v')$. By \Cref{lem: ramsey complemenntary edge implies far}, this implies that $\dist_{T_j}(v,v')\geq 2r-d$. Since $\dist_{T_j}(u,v)+\dist_{T_j}(u,v')\geq\dist_{T_j}(v,v')$, without loss of generality, we can assume that $\dist_{T_j}(u,v)\geq \frac{r}{2}$. Note that $\dist_{G}(u,v)\leq||u-v||_1\leq d||u-v||_{\infty}\leq d$. Therefore, 
\[\min_{1\le i\le k}\set{\max_{v\in V(G)\setminus\set{u}}\frac{\dist_{T_i}(u,v)}{\dist_{G}(u,v)}}=\Omega\bigg(\frac{r}{d}\bigg)=\Omega\bigg(\frac{n^{\frac{1}{k}}}{k}\bigg).\]

\appendix

\section{A $(O(1),2)$-Tree Cover for Grid}
\label{apd: two trees}
Consider the grid $\{x\in\mathbb{Z}^2:||x||_{\infty}\leq\sqrt{n}\}$ on $O(n)$ vertices. In this section, we present the construction of a size-two tree cover $\{T_1,T_2\}$ that covers the grid with distortion $O(1)$.
A similar construction (with much detailed presentation) is given in two concurrent works \cite{le2025coveringeuclideanplanepair,bikeev2025treecoverssize2}.

\begin{figure}[h]
	\centering
	\subfigure[The recursive construction of $T_1$: Tree edges are shown in blue. $c$ is the center of $T_1$, $c_1$ is the center of its northwest subtree, $c_2$ is the center of the southwest subtree of the $c_1$-tree. The path from $c$ to a vertex $v$ in the $c_2$ tree is shown in orange.]
	{
		\scalebox{0.18}{\includegraphics{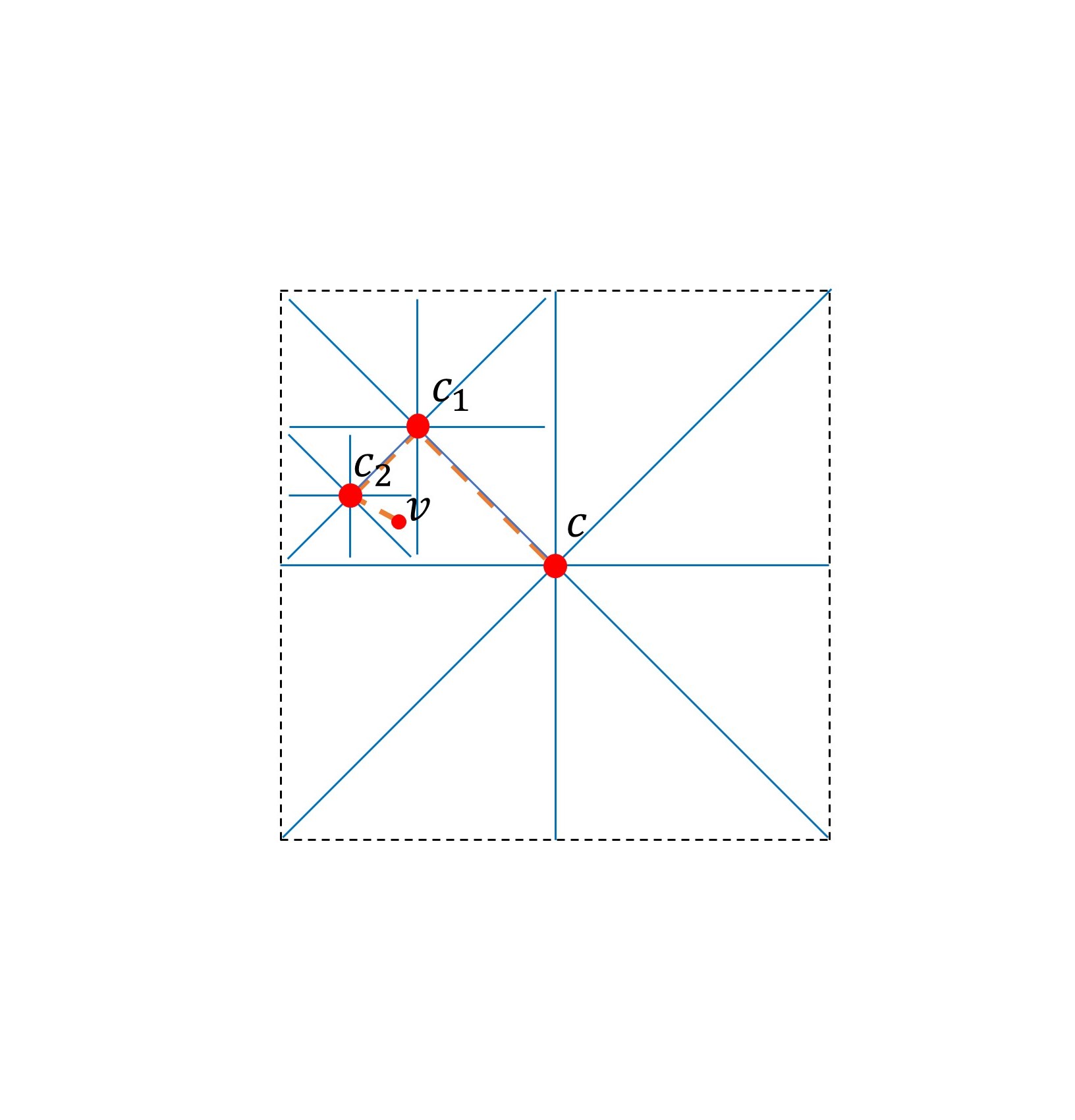}}}
        \hspace{0.2cm}
        \subfigure[The recursive construction of $T_2$: Tree edges are shown in blue. $c$ is the center of $T_2$, $c_1$ is the center of its north subtree, $c_2$ is the center of the west subtree of the $c_1$-tree. The path from $c$ to a vertex $v$ in the $c_2$ tree is shown in orange. Gray lines form a tilted larger grid.]
	{
		\scalebox{0.18}{\includegraphics{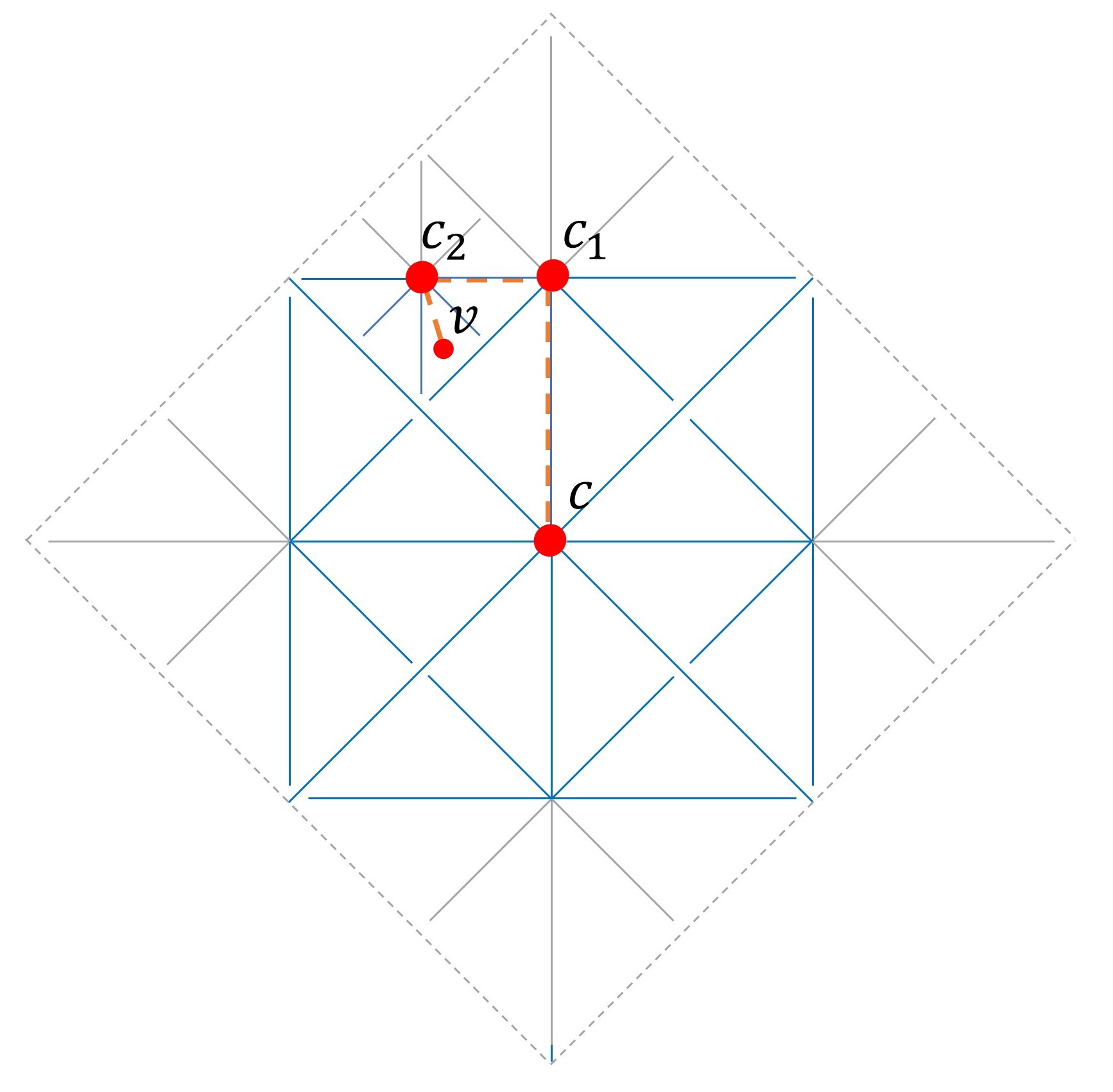}}}
	\caption{The construction of $T_1$ and $T_2$.\label{fig: grid_2trees_construction}}
\end{figure}

We construct $T_1$ as follows. We view $T_1$ as being rooted at the vertex $(0,0)$ and consisting of $8$ subtrees below $c$. Four of them are paths: north path $\big( (0,0),(0,1),\ldots,(0,\sqrt{n}) \big)$, south path: $\big( (0,0),\ldots,(0,-\sqrt{n}) \big)$, the east path $\big( (0,0),\ldots,(\sqrt{n},0) \big)$, and the west path: $\big( (0,0),\ldots,(-\sqrt{n},0) \big)$.
The other remaining are subtrees: the northwest tree, the northeast tree, the southwest tree, and the southeast tree, containing all vertices in different orthants (excluding the points on the $x$ and $y$ axis), respectively. The northwest tree, centered at $(-\frac{\sqrt{n}}{2},\frac{\sqrt{n}}{2})$, is defined recursively (developing $8$ subtrees from $(-\frac{\sqrt{n}}{2},\frac{\sqrt{n}}{2})$, four paths and four subtrees further towards different directions, etc), and is connected to the root $(0,0)$ by the single edge $\big((-1,1),(0,0)\big)$.
See \Cref{fig: grid_2trees_construction} for an illustration. 

The recursive construction naturally defines a quad-tree structure of the grid: layer $1$ is the whole grid; layer $2$ consist of $4$ ``areas'', corresponding to the vertex sets of the northwest, northeast, southwest, and southeast subtrees (vertices on four paths are excluded from these regions).
Each of these areas are further divided into two ``regions'' along its diagonal, so the second layer contains a total $8$ regions. See \Cref{fig: grid_2trees_proof} for an illustration. Deeper layers and its regions are defined recursively in a similar way.
We use the following simple property.

\begin{observation}
\label{obs: to center}
For each center $c$, and for each vertex $v$ in the subtree centered at $c$, the $c$-$v$ grid distance is preserved up to factor $O(1)$ in tree $T_1$. In particular, the grid distances from all vertices to the layer-$1$ center $(0,0)$ is preserved within factor $O(1)$ in $T_1$.
\end{observation}


$T_2$ is constructed as follows. We start by viewing the original grid as a subset of a larger grid rotated by $\pi/4$ (a "tilted grid"). We then construct a tree $T^*$ recursively on this tilted grid, in the same way as $T_1$, and finally restrict this tree onto the original grid (it is easy to verify that $T^*$ induces a connected subtree on vertices of the original grid). This restricted tree is our $T_2$.
See \Cref{fig: grid_2trees_construction} for an illustration. It can be shown that \Cref{obs: to center} also holds for $T_2$.


\begin{figure}[h]
	\centering
	\subfigure[Each component can be partitioned into eight regions. $u,v$ lie in non-adjacent regions.]
	{
		\scalebox{0.145}{\includegraphics{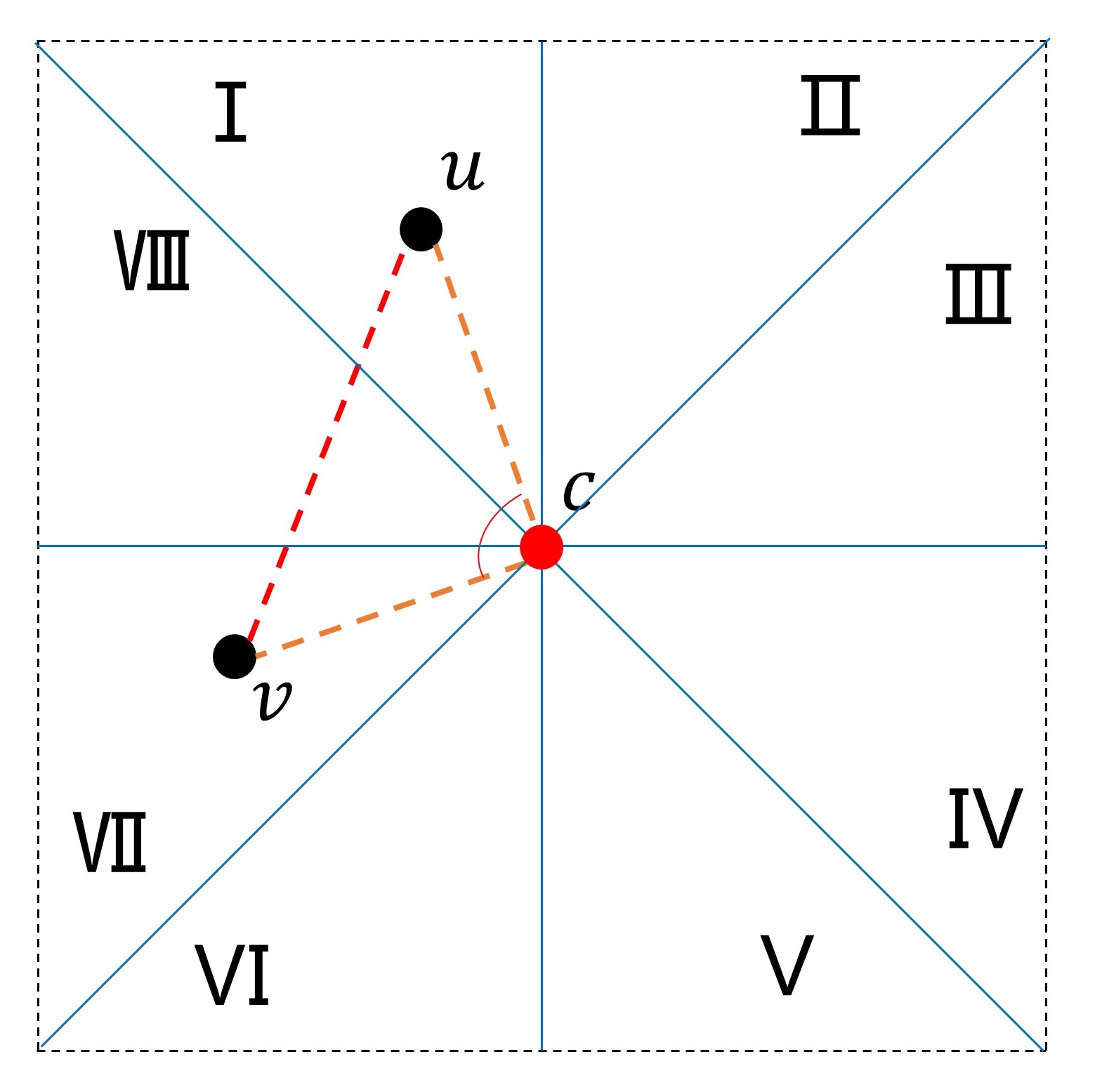}}}
        \hspace{0.3cm}
        \subfigure[Region pairs that are diagonally neighboring.]
	{
		\scalebox{0.16}{\includegraphics{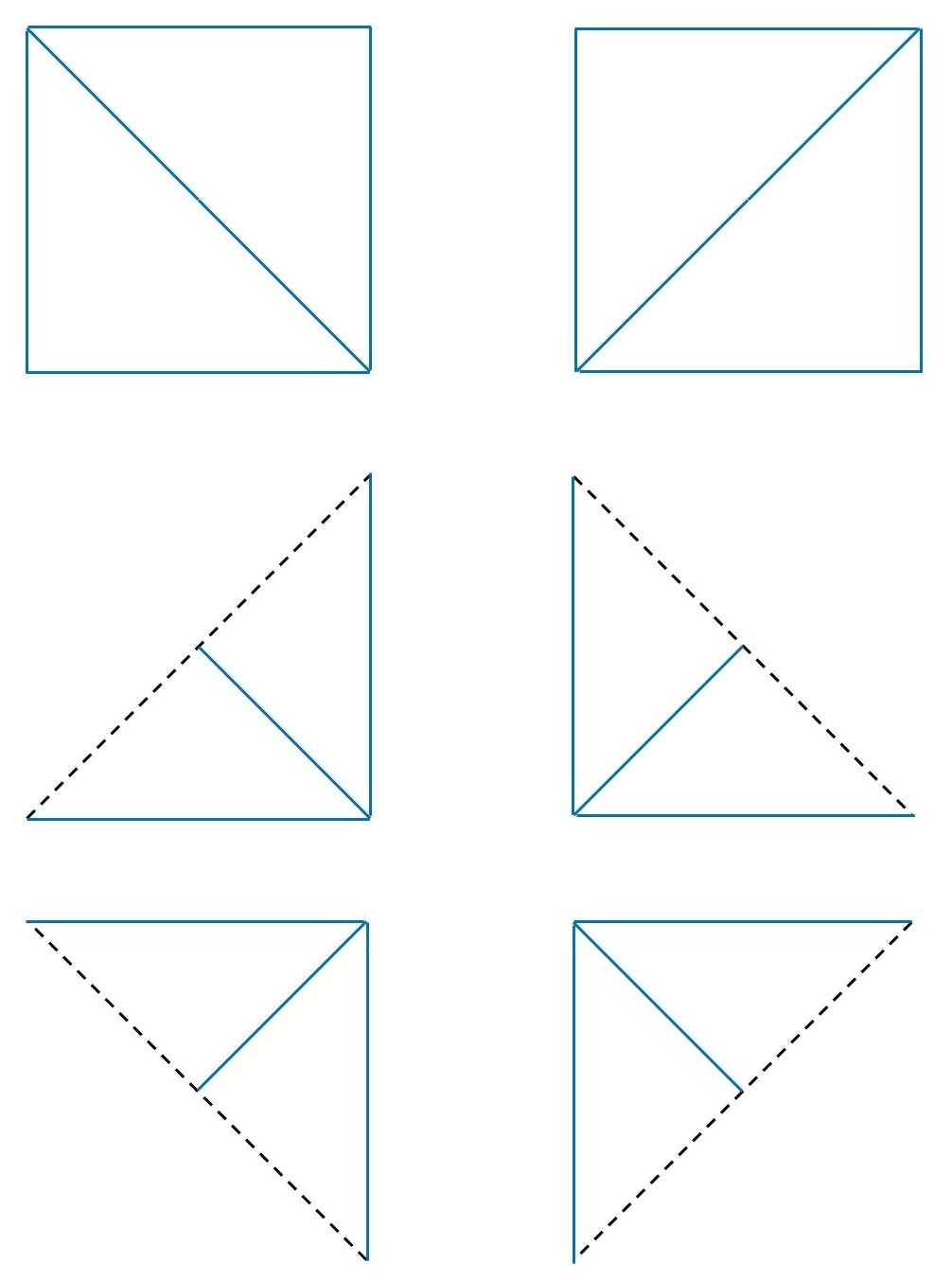}}}
        \hspace{0.3cm}
        \subfigure[Region pairs that are orthogonally neighboring.]
	{
		\scalebox{0.16}{\includegraphics{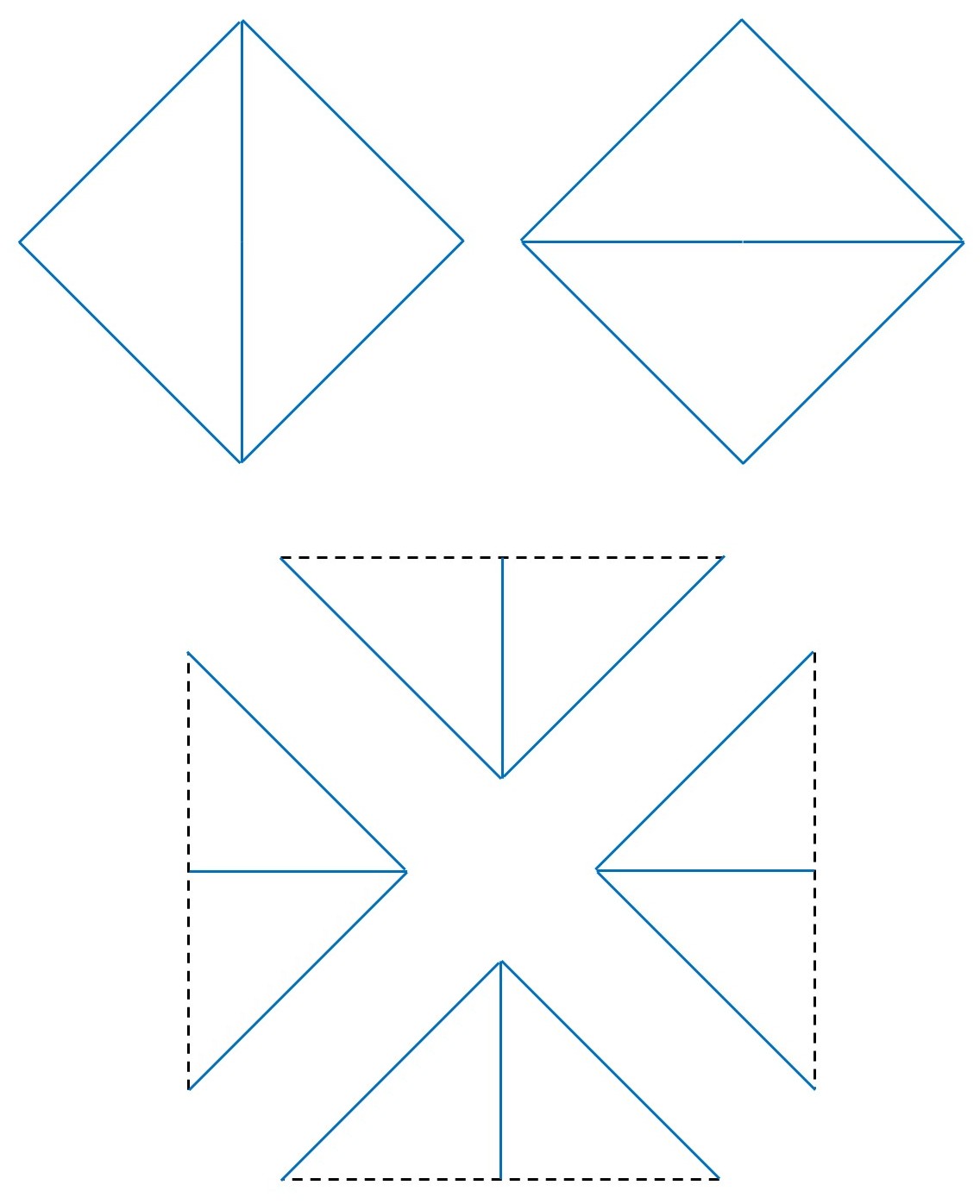}}}
	\caption{Eight regions partitioned from a component, diagonal and orthogonal neighboring.\label{fig: grid_2trees_proof}}
\end{figure}

We now show that the two trees $T_1,T_2$ constructed above cover the grid with distortion $O(1)$.

Consider a pair $u,v$ of vertices in the grid. Let $c$ be the center of largest-indexed layer in $T_1$ that has an area containing both $u$ and $v$. 
Recall that the $c$-tree is divided into $8$ regions.
We distinguish between the following cases.

\paragraph{Case 1. $u,v$ lie in non-adjacent regions} (e.g., regions I and VII  in \Cref{fig: grid_2trees_proof}).
Let $\theta$ denote the angle between vectors $\vec{uc}$ and $\vec{vc}$, then $\pi/4\leq\theta \le 7\pi/4$. Denote by $C$ the constant hidden in \Cref{obs: to center}, by basic trigonometry,
\begin{align*}
    d(u,v)&=\sqrt{d^2(c,u)+d^2(c,v)-2\cos\theta\cdot d(c,u)d(c,v)}\\
    &\geq\sqrt{d^2(c,u)+d^2(c,v)-\sqrt{2} d(c,u)d(c,v)}\\
    &\geq\frac{\sqrt{2-\sqrt{2}}}{2}\cdot \bigg(d(c,u)+d(c,v)\bigg).
\end{align*}
Denote $C'=\frac{2C}{\sqrt{2-\sqrt{2}}}$. Then the tree distance in $T_1$ between $u$ and $v$ satisfies \begin{align*}
    d_{T_1}(u,v)&=d_{T_1}(u,c)+d_{T_1}(c,v)\leq C\cdot(d(u,c)+d(c,v))\\
    &\leq\frac{2C}{\sqrt{2-\sqrt{2}}}\cdot d(u,v)=C'\cdot d(u,v),
\end{align*}
where $d(u,v)$ denotes the grid distance between $u$ and $v$.

\paragraph{Case 2. $u,v$ lie in adjacent regions.} There are two possibilities: diagonally neighboring (e.g., regions I and VIII) or orthogonally neighboring (e.g., regions I and II). In the first case, observe that $u$ and $v$ belong to the same area in the next layer of $T_1$, a contradiction to the assumption that $c$ is the center of largest-indexed layer in $T_1$ that has an area containing both $u$ and $v$. In the second case, they belong to the same area in the next layer of $T_2$. We may then zoom in that layer, further distinguish between the cases where they lie in adjacent or non-adjacent regions, and repeat our analysis.

To sum up, for a pair $u,v$ of vertices: either
\begin{itemize}
    \item they lie in non-adjacent regions at some layer, in which case by similar analysis in Case 1 we can show that their grid distance is preserved up to factor $O(1)$; or
    \item they lie in the same area in the last layer of $T_1$ or $T_2$, which is a subtree of $O(1)$ size and therefore preserves grid distances within $O(1)$ factor.
\end{itemize} 

\bibliographystyle{alpha}
\bibliography{main_arxiv}
\end{document}